\documentclass[11pt,letterpaper]{article}

\usepackage[utf8]{inputenc}
\usepackage{amsmath}
\usepackage{amstext}
\usepackage{amsthm}
\usepackage{bbold}
\usepackage{bm}
\usepackage{bbm}
\usepackage{colonequals}
\usepackage{comment}
\usepackage[dvips,letterpaper,margin=1.0in]{geometry}

\newtheorem{theorem}{Theorem}[section]
\newtheorem{lemma}[theorem]{Lemma}

\newtheorem{example}[theorem]{Example}
\newtheorem{proposition}[theorem]{Proposition}
\newtheorem{corollary}[theorem]{Corollary}
\newtheorem{conjecture}[theorem]{Conjecture}
\newtheorem{definition}[theorem]{Definition}
\newtheorem{remark}[theorem]{Remark}

\newcommand{\RR}{\mathbb{R}}
\newcommand{\QQ}{\mathbb{Q}}
\newcommand{\NN}{\mathbb{N}}
\newcommand{\ZZ}{\mathbb{Z}}

\newcommand{\PP}{\mathbb{P}}
\newcommand{\EE}{\mathbb{E}}

\newcommand{\One}{\mathbbm{1}}

\newcommand{\Ex}{\mathop{\mathbb{E}}}  

\newcommand{\sC}{\mathcal{C}}
\newcommand{\sD}{\mathcal{D}}

\newcommand{\sF}{\mathcal{F}}

\newcommand{\sN}{\mathcal{N}}
\newcommand{\sO}{\mathcal{O}}

\newcommand{\sS}{\mathcal{S}}

\newcommand{\sX}{\mathcal{X}}

\newcommand{\bC}{\bm C}

\newcommand{\bM}{\bm M}

\newcommand{\bP}{\bm P}
\newcommand{\bQ}{\bm Q}

\newcommand{\bW}{\bm W}

\newcommand{\bY}{\bm Y}

\newcommand{\bg}{\bm g}

\newcommand{\what}{\widehat}

\newcommand{\bv}{\bm v}

\newcommand{\bx}{\bm x}
\newcommand{\by}{\bm y}

\newcommand{\la}{\langle}
\newcommand{\ra}{\rangle}

\DeclareMathOperator{\Tr}{Tr}

\newcommand{\iid}{\mathsf{iid}}
\newcommand{\sgn}{\mathrm{sgn}}
\newcommand{\trunc}{\mathrm{trunc}}
\newcommand{\PCA}{\mathsf{PCA}}
\newcommand{\GOE}{\mathsf{GOE}}
\newcommand{\Wishart}{\mathsf{Wishart}}

\makeatletter
\renewcommand*{\@fnsymbol}[1]{\ensuremath{\ifcase#1\or *\or \ddagger\or
    \mathsection\or \mathparagraph\or \|\or **\or \dagger\dagger
    \or \ddagger\ddagger \else\@ctrerr\fi}}
\makeatother

\title{Computational Hardness of Certifying Bounds on Constrained PCA Problems}
\author{Afonso S.\ Bandeira, Dmitriy Kunisky, Alexander S.\ Wein}

\usepackage{authblk}
\author[1,2]{Afonso S.\ Bandeira\thanks{Email: \textit{bandeira@cims.nyu.edu}. Partially supported by NSF grants DMS-1712730 and DMS-1719545, and by a grant from the Sloan
Foundation.}}
\author[1]{Dmitriy Kunisky\thanks{Email: \textit{kunisky@cims.nyu.edu}. Partially supported by NSF grants DMS-1712730 and DMS-1719545.}}
\author[1]{Alexander S.\ Wein\thanks{Email: \textit{awein@cims.nyu.edu}. Partially supported by NSF grant DMS-1712730 and by the Simons Collaboration on Algorithms and Geometry.}}

\affil[1]{Department of Mathematics, Courant Institute of Mathematical Sciences, New York University, New York, NY, USA}
\affil[2]{Center for Data Science, New York University, New York, NY, USA}

\date{April 5, 2019}

\begin{document}

\maketitle

\begin{abstract}
Given a random $n \times n$ symmetric matrix $\bm W$ drawn from the Gaussian orthogonal ensemble (GOE), we consider the problem of certifying an upper bound on the maximum value of the quadratic form $\bm x^\top \bm W \bm x$ over all vectors $\bm x$ in a constraint set $\mathcal{S} \subset \mathbb{R}^n$. For a certain class of normalized constraint sets $\mathcal{S}$ we show that, conditional on certain complexity-theoretic assumptions, there is no polynomial-time algorithm certifying a better upper bound than the largest eigenvalue of $\bm W$.
A notable special case included in our results is the hypercube $\sS = \{ \pm 1 / \sqrt{n}\}^n$, which corresponds to the problem of certifying bounds on the Hamiltonian of the Sherrington-Kirkpatrick spin glass model from statistical physics.

Our proof proceeds in two steps. First, we give a reduction from the detection problem in the \emph{negatively-spiked Wishart model} to the above certification problem. We then give evidence that this Wishart detection problem is computationally hard below the classical spectral threshold, by showing that no low-degree polynomial can (in expectation) distinguish the spiked and unspiked models.
This method for identifying computational thresholds was proposed in a sequence of recent works on the sum-of-squares hierarchy, and is believed to be correct for a large class of problems.
Our proof can be seen as constructing a distribution over symmetric matrices that appears computationally indistinguishable from the GOE, yet is supported on matrices whose maximum quadratic form over $\bm x \in \mathcal{S}$ is much larger than that of a GOE matrix.
\end{abstract}

\newpage

\section{Introduction}

An important phenomenon in the study of the computational characteristics of random problems is the appearance of \emph{statistical-to-computational gaps}, wherein a problem may be solved by an inefficient algorithm---typically a brute-force search---but empirical evidence, heuristic formal calculations, and negative results for classes of powerful algorithms all suggest that the same problem cannot be solved by any algorithm running in polynomial time.
Many examples of this phenomenon arise from Bayesian estimation tasks, in which the goal is to recover a \emph{signal} from noisy observations.
Bayesian problems exhibiting statistical-to-computational gaps in certain regimes include graph problems such as community detection \cite{block-model}, estimation for models of structured matrices and tensors \cite{mmse-rank-1,tensor-pca}, statistical problems arising from imaging and microscopy tasks \cite{synch-amp,mra-het}, and many others.
A different family of examples comes from random optimization problems that are \emph{signal-free}, where there is no ``planted'' structure to recover; rather, the task is simply to optimize a random objective function as effectively as possible.
Notable instances of problems of this kind that exhibit statistical-to-computational gaps include finding a large clique in a random graph \cite{jerrum-clique}, finding a large submatrix of a random matrix \cite{gamarnik-submatrix}, or finding an approximate solution to a random constraint satisfaction problem \cite{coja-barriers}.

In this paper, we study a problem from the latter class, the problem of maximizing a quadratic form $\bx^\top \bW \bx$ over a constraint set $\bx \in \sS \subset \RR^n$, where $\bW$ is a random matrix drawn from the Gaussian orthogonal ensemble,\footnote{Gaussian orthogonal ensemble (GOE): $\bW$ is symmetric with $W_{ij} = W_{ji} \sim \sN(0, 1/n)$ for $i \neq j$ and $W_{ii} \sim \sN(0, 2/n)$ independently.} $\bW \sim \GOE(n)$.
Unlike previous works that have studied whether an efficient algorithm can \emph{optimize} and find $\bx = \bx(\bW)$ that achieves a large objective value, we study whether an efficient algorithm can \emph{certify} an upper bound on the objective over all $\bx \in \sS$.
In the notable case of the \emph{Sherrington-Kirkpatrick (SK) Hamiltonian} \cite{sk,pan-sk}, where $\sS = \{ \pm 1 / \sqrt{n}\}^n$, while there is an efficient algorithm believed to optimize arbitrarily close to the true maximum \cite{mon-sk}, we show that, conditional on the correctness of the \emph{low-degree likelihood ratio method} recently developed in the sum-of-squares literature \cite{sos-clique,HS-bayesian,sos-power,sam-thesis} (which we explain in Section~\ref{sec:low-deg}), there is \emph{no} efficient algorithm to certify an upper bound that improves on a simple spectral certificate.
Thus, the certification task for this problem exhibits a statistical-to-computational gap, while the optimization task does not.

\paragraph{Signal-free random optimization problems.}

The general task we will be concerned with is the optimization of a random function,
\begin{equation}
    \begin{array}{ll}
    \text{maximize} & f_{\omega}(\bx) \\
    \text{subject to} & \bx \in \sS \\
    \text{where} & \omega \sim \PP.
    \end{array}
\end{equation}
\noindent
Sometimes, such a task arises in statistical estimation as a likelihood maximization, where the random function $f_\omega(\bx)$ is the likelihood of an observed dataset $\omega$ for a given parameter value $\bx$.
But the same formal task, stripped of this statistical origin, is still common: in statistical physics, random functions arise in models of magnetism in disordered media; in optimization, random functions encode uncertainty in the parameters of a problem; and in theoretical computer science, random instances of algorithmic tasks describe the average-case rather than worst-case computational characteristics of a problem.
Below, we review a well-studied example showing the connection between a prominent statistical estimation problem and a related signal-free random optimization problem.

\begin{example}
    \label{ex:sk}
    Consider the \emph{Rademacher-spiked Wigner model}, a family of probability distributions $\PP_{\lambda, \bx}$ indexed by $\bx \in \{ \pm 1 / \sqrt{n} \}^n$ and $\lambda > 0$.
    Letting $\bW \sim \GOE(n)$, $\PP_{\lambda, \bx}$ is the law of the random matrix $\lambda \bx\bx^\top + \bW$.
    For $\lambda$ fixed, the log-likelihood of $\bx$ for some observed data $\bY \in \RR^{n \times n}$ is
    \[ \log \PP_{\lambda, \bx}[\bY] = -\frac{n}{4}\|\bY - \lambda \bx\bx^\top \|_F^2 = c_{\bY} + \frac{\lambda n}{2} \bx^\top \bY \bx, \]
    where $c_{\bY}$ depends on $\bY$ but not on $\bx$.

    Thus, drawing $\bY \sim \PP_{\lambda, \bx^\star}$ and maximizing the likelihood gives a random optimization problem,
    \begin{equation}
    \begin{array}{ll}
    \text{maximize} & \bx^\top \bY \bx \\
    \text{subject to} & \bx \in \{\pm 1 / \sqrt{n}\}^n \\
    \text{where} & \bY \sim \PP_{\lambda, \bx^\star}.
    \end{array}
    \end{equation}
    Success in the associated estimation problem corresponds to recovering $\bx^\star$ as the solution to this problem; the ``overlap'' $\la \bx, \bx^\star \ra$ is often used as a quantitative measure of success in this task.

    A natural ``signal-free'' version of this problem arises by setting $\lambda = 0$.
    In this case, note that $\PP_{0, \bx^\star} = \PP_0$ does not actually depend on $\bx^\star$, leaving us with the optimization
    \begin{equation}
    \label{eq:sk-intro}
    \begin{array}{ll}
    \text{maximize} & \bx^\top \bW \bx \\
    \text{subject to} & \bx \in \{\pm 1 / \sqrt{n}\}^n \\
    \text{where} & \bW \sim \GOE(n).
    \end{array}
    \end{equation}
    Up to scaling and a change of sign, this task is the same as that of identifying the ground state configuration or energy in the Sherrington-Kirkpatrick (SK) spin glass model \cite{sk,pan-sk}. For this reason we refer to~\eqref{eq:sk-intro} as the \emph{SK problem}.
    Note that there is no ``planted'' solution $\bx^\star$ with respect to which we may measure an algorithm's performance; rather, the quality of the $\bx$ an algorithm obtains is measured only by the value of $\bx^\top \bW \bx$.
\end{example}

\paragraph{Computational tasks: optimization vs.\ certification.}

Let us contrast two computational tasks of interest for a given optimization problem.
The first, most obvious task is that of \emph{optimization}, producing an algorithm computing $\mathsf{alg}_{\mathsf{opt}}: \Omega \to \sS$ such that $f_{\omega}(\mathsf{alg}_{\mathsf{opt}}(\omega))$ is as large as possible (say, in expectation, or with high probability as the size of the problem diverges).

Another task is that of \emph{certification}, producing instead an algorithm computing a number $\mathsf{alg}_{\mathsf{cert}}: \Omega \to \RR$, such that \emph{for all} $\omega \in \Omega$ and all $x \in \mathcal{S}$ we have $f_{\omega}(x) \leq \mathsf{alg}_{\mathsf{cert}}(\omega)$.
The main additional challenge of certification over optimization is that $\mathsf{alg}_{\mathsf{cert}}$ must produce a valid upper bound on $f_{\omega}$ for \emph{every} possible value of the data $\omega \in \Omega$, no matter how unlikely $\omega$ is to occur under $\PP$.
Subject to this requirement, we seek to minimize $\mathsf{alg}_{\mathsf{cert}}(\omega)$ (again, in a suitable probabilistic sense when $\omega \sim \PP$).
\emph{Convex relaxations} are a common approach to certification, where $\sS$ is relaxed to a convex superset $\sS^\prime \supset \sS$ admitting a sufficiently simple description that it is possible to optimize exactly over $\sS^\prime$ using convex optimization.

If $\bx^\star = \bx^\star(\omega)$ is the true maximizer of $f_{\omega}$, then for any pair of optimization and certification algorithms as above, we have
\begin{equation}
    f_{\omega}(\mathsf{alg}_{\mathsf{opt}}(\omega)) \leq f_{\omega}(\bx^\star) \leq \mathsf{alg}_{\mathsf{cert}}(\omega).
\end{equation}
Thus, in the case of a maximization problem, optimization algorithms approximate the true value $f_{\omega}(\bx^\star)$ from below, while certification algorithms approximate it from above.
We are then interested in how tight either inequality for random problems in growing dimension. Of course, we can achieve ``perfect'' optimization and certification by exhaustive search over all $\bx \in \sS$, but we are interested only in computationally efficient algorithms.

To make these definitions concrete, we review an instance of each type of algorithm for the problem of Example~\ref{ex:sk}.

\setcounter{theorem}{0}
\begin{example}[Continued]
    In the SK problem~\eqref{eq:sk-intro}, two related \emph{spectral algorithms} give simple examples of algorithms for both optimization and certification.

    For certification, writing $\lambda_{\max}$ for the largest eigenvalue of $\bW$, we may use the bound
    \begin{equation}
    \label{eq:spectral-cert}
    \bx^\top \bW \bx \leq \lambda_{\max} \cdot \|\bx\|^2 = \lambda_{\max} \approx 2
    \end{equation}
    for all $\bx \in \{ \pm 1 / \sqrt{n}\}^n$, whereby $\lambda_{\max}$ is a certifiable upper bound on~\eqref{eq:sk-intro}.
    From classical random matrix theory (see, e.g., \cite{AGZ-book}), it is known that $\lambda_{\max} \approx 2$ as $n \to \infty$.

    For optimization, for $\bv_{\max}$ the eigenvector of $\lambda_{\max}$, we may take $\bx = \bx(\bW) \colonequals \sgn(\bv_{\max}) / \sqrt{n}$ where $\sgn$ denotes the $\{\pm 1\}$-valued sign function, applied entrywise.
    The vector $\bv_{\max}$ is distributed as an isotropically random unit vector in $\RR^n$, so the quality of this solution may be computed as
    \begin{equation}
    \bx^\top \bW \bx = \lambda_{\max} \cdot \la \bx, \bv_{\max} \ra^2 + O\left(\frac{1}{\sqrt{n}}\right) = \lambda_{\max} \cdot \frac{\|\bv_{\max}\|_1^2}{n} + O\left(\frac{1}{\sqrt{n}}\right) \approx \frac{4}{\pi} \approx 1.2732
    \end{equation}
    with high probability as $n \to \infty$.
    (The error in the first equation is obtained as $\sum_{i}\lambda_i \la \bv_i, \bx\ra^2 \approx \frac{1}{n}\Tr(\bW)(1 - \la \bv_{\max}, \bx \ra^2)$, where the sum is over all eigenvectors $\bv_i$ except $\bv_{\max}$. This analysis appeared in \cite{ALR-sk}, an early rigorous mathematical work on the SK model.)

    On the other hand, deep results of statistical physics imply that the true optimal value approaches
    \begin{equation}
        \bx^{\star^\top} \bW \bx^{\star} \approx 2\mathsf{P}_* \approx 1.5264
    \end{equation}
    as $n \to \infty$, where the constant $\mathsf{P}_*$ is expressed via the celebrated Parisi formula for the free energy of the SK model \cite{par-frsb,pan-sk,tal-parisi}.
    The approximate value we give above was estimated with numerical experiments in previous works (see, e.g., \cite{par-sk-num,cris-sk}).
\end{example}
\noindent
The recent result of \cite{mon-sk} implies, assuming a widely-believed conjecture from statistical physics, that for any $\varepsilon > 0$ there exists a polynomial-time optimization algorithm achieving with high probability a value of $2\mathsf{P}_* - \varepsilon$ on the SK problem.
This work builds on that of \cite{ABM-crem,sub-frsb}, and these works taken together formalize the heuristic idea from statistical physics that optimization is tractable for certain optimization problems exhibiting \emph{full replica symmetry breaking}.
On the other hand, there are few results addressing the SK \emph{certification} problem.
The only previous work we are aware of in this direction is \cite{MS-sdp}, where a simple semidefinite programming relaxation is shown to achieve the same value as the spectral certificate~\eqref{eq:spectral-cert}.

\paragraph{Our contributions.}
The main result of this paper, which we now state informally, shows that for the SK certification problem, the simple spectral certificate~\eqref{eq:spectral-cert} is optimal. See Corollary~\ref{cor:pca-hard} for the formal statement.
\begin{theorem}[Informal]\label{thm:main-informal}
Conditional on the correctness of the low-degree likelihood ratio method (see Section~\ref{sec:low-deg}), for any $\varepsilon > 0$, there is no polynomial-time algorithm that certifies the upper bound $2-\varepsilon$ on the SK problem~\eqref{eq:sk-intro} with probability $1-o(1)$ as $n \to \infty$.
\end{theorem}

\noindent
Theorem~\ref{thm:main-informal} reveals a striking gap between optimization and certification: it is possible to efficiently give a tight lower bound on the maximum objective value by exhibiting a solution $\bx$, but impossible to efficiently give a tight upper bound. In other words, an algorithm can efficiently find a near-optimal solution, but cannot be sure that it has done so.
The same result also holds for a wide variety of constraints other than $\bx \in \{\pm 1/\sqrt{n}\}^n$ (see Corollary~\ref{cor:pca-hard}). Due to the high-dimensional setting of the problem, we expect the value of a certification algorithm to concentrate tightly; thus we also expect Theorem~\ref{thm:main-informal} to still hold if $1-o(1)$ is replaced by any positive constant.

Our result has important consequences for convex programming. A natural approach for optimizing the SK problem~\eqref{eq:sk-intro} would be to use a convex programming relaxation such as a semidefinite program based on the sum-of-squares hierarchy \cite{shor-sos,par-sos,las-sos}. Such a method would relax the constraints of the SK problem to weaker ones for which the associated optimization problem can be solved efficiently. One can either hope that the relaxation is \emph{tight} and gives a valid solution $\bx \in \{\pm 1/\sqrt{n}\}^n$ (with high probability), or use a \emph{rounding} procedure to extract a valid solution from the relaxation. The optimal value of any convex relaxation of~\eqref{eq:sk-intro} provides an upper bound on the optimal value of~\eqref{eq:sk-intro} and therefore gives a certification algorithm. Thus Theorem~\ref{thm:main-informal} implies that (conditional on the correctness of the low-degree likelihood ratio method) no polynomial-time convex relaxation of~\eqref{eq:sk-intro} can have value $\le 2-\varepsilon$ (resolving a question posed by \cite{JKR-meanfield}) and in particular cannot be tight.
As a result, we expect that natural relax-and-round approaches for optimization should fail to find a solution of value close to $2\mathsf{P}_*$.
This would suggest a fundamental weakness of convex programs: even the most powerful convex programs (such as sum-of-squares relaxations) seem to fail to optimize~\eqref{eq:sk-intro}, even though other methods succeed (namely, the message-passing algorithm of \cite{mon-sk}).\footnote{In contrast, simple rounded convex relaxations are believed to approximate many similar problems optimally in the worst-case (rather than average-case) setting \cite{KKMO-csp}.} An explanation for this suboptimality is that convex relaxations are actually solving a fundamentally harder problem: certification.

\paragraph{Related work.}

The SK problem is not the first known instance of a problem where perfect optimization is tractable but perfect certification appears to be hard. One example comes from random constraint satisfaction problems (CSPs).

\begin{example}
    In random $\mathsf{MAX}$-$3\mathsf{SAT}$, we draw a uniformly random $n$-variable $m$-clause 3-CNF formula and maximize $s_{\bC}(\bx)$, the number of satisfied clauses, over $\bx \in \{0,1\}^n$.

    If $m/n \to \infty$ as $n \to \infty$, the optimal value $\max_{\bx} s_{\bC}(\bx)$ is $(7/8 + o(1)) m$ with probability $1-o(1)$ \cite{coja-refutation}. This is achieved by the trivial optimization algorithm that chooses a uniformly random assignment $\bx$. On the other hand, sum-of-squares lower bounds suggest that it is hard to certify even $s_{\bC}(\bx) < m$ unless $m \gg n^{3/2}$ \cite{sos-refute}.
\end{example}

\noindent Prior work has used \emph{sum-of-squares lower bounds} to argue for hardness of certification in problems such as random CSPs \cite{sos-refute}, planted clique \cite{DM-clique,MPW-clique,sos-clique}, tensor injective norm \cite{tensor-pca,sos-power}, graph coloring \cite{BKM-coloring}, community detection in hypergraphs \cite{hypergraph}, and others. These results prove that the sum-of-squares hierarchy (at some degree) fails to certify. If sum-of-squares fails at every constant degree (e.g., \cite{sos-clique,sos-refute,sos-power}), this suggests that all polynomial-time algorithms should also fail. In our case, it appears difficult to prove sum-of-squares lower bounds for the SK problem, so we instead take a new approach based on a related heuristic for computational hardness, which we explain in the next section.

\paragraph{Overview of techniques.}

The proof of Theorem~\ref{thm:main-informal} has two parts. First, we give a reduction from hypothesis testing in the \emph{negatively-spiked Wishart model} \cite{J-spike,BBP,BS-wishart,PWBM-contig} to the SK certification problem. We then use a method introduced in the sum-of-squares literature based on the \emph{low-degree likelihood ratio} \cite{HS-bayesian,sos-power,sam-thesis} to give evidence that detection in the negatively-spiked Wishart model is computationally hard (in the relevant parameter regime).

In the spiked Wishart model, we observe either $N$ i.i.d.\ samples $\by_1, \dots, \by_N \sim \sN(0,\bm I_n)$, or $N$ i.i.d.\ samples $\by_1, \dots, \by_N \sim \sN(0,\bm I_n + \beta \bx \bx^\top)$ where the ``spike'' $\bx \in \{\pm 1/\sqrt{n}\}^n$ is a uniformly random hypercube vector, and $\beta \in (-1,\infty)$. The goal is to distinguish between these two cases with probability $1-o(1)$ as $n \to \infty$. In the \emph{negatively-spiked} ($\beta < 0$) case with $\beta \approx -1$, this task amounts to deciding whether there is a hypercube vector $\bx \in \{\pm 1/\sqrt{n}\}^n$ that is nearly orthogonal to all of the samples $\by_i$. When $N = \Theta(n)$, a simple spectral method succeeds when $\beta^2 > n/N$ \cite{BBP,BS-wishart}, and we expect the problem to be computationally hard when $\beta^2 < n/N$.

Let us now intuitively explain the relation between the negatively-spiked Wishart model and the SK certification problem. Suppose we want to certify that
$$\mathsf{SK}(\bW) \colonequals \max_{\bx \in \{\pm 1/\sqrt{n}\}^n} \bx^\top \bW \bx \le 2 - \varepsilon$$
where $\bW \sim \GOE(n)$, for some small constant $\varepsilon > 0$. Since the eigenvalues of $\bW$ follow the semicircle distribution on $[-2,2]$ \cite{wigner-semicircle}, we need to certify that the top $\delta n$-dimensional eigenspace of $\bW$ does not (approximately) contain a hypercube vector, for some small $\delta > 0$ depending on $\varepsilon$. In particular, we need to distinguish a uniformly random $\delta n$-dimensional subspace (the distribution of the actual top eigenspace of $\bW \sim \GOE(n)$) from a $\delta n$-dimensional subspace that contains a hypercube vector. Equivalently, taking orthogonal complements, we need to distinguish a uniformly random $(1-\delta) n$-dimensional subspace from a $(1-\delta) n$-dimensional subspace that is orthogonal to a hypercube vector. This is essentially the problem of detection in the negatively-spiked Wishart model with $\beta \approx -1$ and $N = (1-\delta)n$, and these parameters lie in the ``hard regime'' $\beta^2 < n/N$.

Formally, we construct a distribution $\sD(n)$ over $n \times n$ symmetric matrices with $\mathsf{SK}(\bW) \geq 2-\varepsilon/2$ when $\bW \sim \sD(n)$. This $\sD(n)$ also has the property that, conditional on the hardness of the above detection problem, it is computationally hard to distinguish $\bW \sim \sD(n)$ from $\bW \sim \GOE(n)$. The existence of such $\sD(n)$ implies hardness of certification for the SK problem, because if an algorithm could certify that $ \mathsf{SK}(\bW) \le 2-\varepsilon$ when $\bW \sim \GOE(n)$, then it could distinguish $\sD(n)$ from $\GOE(n)$.

The idea of ``planting'' a hidden solution (in our case, a hypercube vector $\bx$) in such a way that it is difficult to detect is referred to as \emph{quiet planting} \cite {quiet-1,quiet-2}. Our quiet planting scheme $\sD(n)$ draws $\bW \sim \GOE(n)$ and then rotates the top eigenspace of $\bW$ to align with a random hypercube vector $\bx$, while leaving the eigenvalues of $\bW$ unchanged. (The more straightforward planting scheme, $\bW + (2-\varepsilon/2) \bx \bx^\top$ with $\bW \sim \GOE(n)$, is not quiet because it changes the largest eigenvalue of $\bW$ \cite{fp}.) The question of how to design optimal quiet planting schemes in general remains an interesting open problem.

The final ingredient in our proof is to argue that detection in the spiked Wishart model is computationally hard below the spectral threshold. We do this through a calculation involving the projection of the likelihood ratio between the ``null'' and ``planted'' distributions of this model onto the subspace of low-degree polynomials. This method may be viewed as an implementation of the intuitive idea that the correct strategy for quiet planting is to match the low-degree moments of the distributions $\sD(n)$ and $\GOE(n)$. We discuss the details of this method further in Section~\ref{sec:low-deg}.

Our results on hardness in the spiked Wishart model may be of independent interest: our calculations indicate that, for a large class of spike priors, no polynomial-time algorithm can successfully distinguish the spiked and unspiked models below the classical spectral threshold \cite{BBP,BS-wishart}, both in the negatively-spiked and positively-spiked regimes.

\section{Background}

\subsection{Probability Theory}

All our asymptotic notation (e.g., $O(\cdot), o(\cdot)$) pertains to the limit $n \to \infty$. We consider parameters of the problem (e.g., $\beta,\gamma,\sX,\sS$) to be held fixed as $n \to \infty$. Thus, the constants hidden by $O(\cdot)$ and $o(\cdot)$ do not depend on $n$ but may depend on the other parameters.

\begin{definition}
    If $(\Omega_n, \sF_n, \PP_n)$ is a sequence of probability spaces, and $A = (A_n)_{n \in \NN}$ is a sequence of events with $A_n \in \sF_n$, then we say $A$ holds \emph{with high probability} if $\lim_{n \to \infty} \PP_n[A_n] = 1$.
\end{definition}

\begin{definition}\label{def:subg}
A real-valued random variable $\pi$ with $\EE[\pi] = 0$ is \emph{subgaussian} if there exists $\sigma^2$ (the \emph{variance proxy}) such that, for all $t \in \RR$, $M(t) \colonequals \EE[\exp(t \pi)] \le \exp(\sigma^2 t^2 / 2)$.
\end{definition}
\noindent
The name \emph{subgaussian} refers to the fact that if $\pi \sim \sN(0, \sigma^2)$, then $M(t) = \exp(\sigma^2 t^2 / 2)$. A random variable with law $\sN(0, \sigma^2)$ is therefore subgaussian. Any bounded centered random variable is also subgaussian: if $\pi \in [a, b]$ almost surely, then $\pi$ is subgaussian with $\sigma^2 = \frac{1}{4}(b - a)^2$ (see, e.g., \cite{rig-notes}).

We next give some background facts from random matrix theory (see, e.g., \cite{AGZ-book}).
\begin{definition}
    The \emph{Gaussian orthogonal ensemble} $\GOE(n)$ is a probability distribution over symmetric matrices $\bW \in \RR^{n \times n}$, under which $W_{ii} \sim \sN(0, 2/n)$ and $W_{ij} \sim \sN(0, 1/n)$ when $i \neq j$, where the entries $W_{ij}$ are independent for distinct pairs $(i, j)$ with $i \leq j$.
\end{definition}
\noindent The choice of variances ensures the following crucial invariance property of $\GOE(n)$.
\begin{proposition}
    \label{prop:goe-invariance}
    For any $\bQ \in \sO(n)$, if $\bW \sim \GOE(n)$, then the law of $\bQ\bW\bQ^\top$ is also $\GOE(n)$.
\end{proposition}
\noindent Our scaling of the entries of $\GOE(n)$ is chosen to ensure a spectrum of constant width.
\begin{proposition}
    \label{prop:goe-spectrum}
    Let $\bW_n \sim \GOE(n)$.
    Then, almost surely, $\lambda_{\min}(\bW_n) \to -2$ and $\lambda_{\max}(\bW_n) \to 2$ as $n \to \infty$.
    In particular, for any $\varepsilon > 0$, $\|\bW_n\| \leq 2 + \varepsilon$ with high probability.
\end{proposition}
\noindent Furthermore, by Wigner's semicircle law \cite{wigner-semicircle}, the empirical distribution of eigenvalues of $\bW_n$ converges weakly to a semicircle distribution supported on $[-2,2]$.

\subsection{Constrained PCA}

\begin{definition}
A \emph{constraint set} is a sequence $\sS = (\sS_n)_{n \in \NN}$ where $\sS_n \subseteq \RR^n$. The \emph{constrained principal component analysis (PCA) problem with constraint set $\sS$}, denoted $\PCA(\sS)$, is
\[ \begin{array}{ll}
    \text{maximize} & \bx^\top \bW \bx \\
    \text{subject to} & \bx \in \sS_n \\
    \text{where} & \bW \sim \GOE(n).
\end{array}  \]
\end{definition}
\noindent We will work only with constraint sets supported on vectors of approximately unit norm.

\begin{example}
    Problems that may be described in the constrained PCA framework include:
    \begin{itemize}
        \item the Sherrington-Kirkpatrick (SK) spin glass model: $\sS_n = \{\pm 1 / \sqrt{n}\}^n$ \cite{sk,pan-sk},
        \item the sparse PCA null model: $\sS_n = \{\bx \in \RR^n \;:\; \|\bx\| = 1, \|\bx\|_0 \le \rho\}$ \cite{info-sparse-pca,mi-rank-1},
        \item the spherical $2p$-spin spin glass model: $\sS_{pn} = \{\bx^{\otimes p} \;:\; \bx \in \RR^n, \,\|\bx\| = 1\}$ \cite{pspin-1,pspin-2},
        \item the positive PCA null model: $\sS_n = \{\bx \in \RR^n \;:\; x_i \ge 0,\, \|\bx\| = 1\}$ \cite{positive-pca}.
    \end{itemize}
\end{example}
\noindent Our results apply to the first two examples: the SK model, and sparse PCA when $\rho = \Theta(n)$.

\begin{definition}
Let $f$ be a (randomized) algorithm that takes a square matrix $\bm W$ as input and outputs a number $f(\bm W) \in \RR$. We say that $f$ \emph{certifies} a value $B$ on $\PCA(\sS)$ if
\begin{enumerate}
    \item for \emph{any} symmetric matrix $\bm W \in \RR^{n \times n}$, $\max_{\bx \in \sS_n} \bx^\top \bm W \bx \le f(\bm W)$, and
    \item if $\bm W_n \sim \GOE(n)$ then $f(\bm W_n) \le B + o(1)$ with high probability.
\end{enumerate}
\end{definition}
\noindent We allow $f$ to be randomized (i.e., it may use randomness in its computations, but the output $B$ must be an upper bound almost surely). We do not expect certification algorithms to require randomness, but it may be convenient, e.g., to randomly initialize an iterative optimization procedure.

\subsection{Spiked Wishart Models}

\begin{definition}
A \emph{normalized spike prior} is a sequence $\sX = (\sX_n)_{n \in \NN}$ where $\sX_n$ is a probability distribution over $\RR^n$, such that if $\bx \sim \sX_n$ then $\|\bx\| \to 1$ in probability as $ n \to \infty$.
\end{definition}

\begin{definition}[Spiked Wishart model]
    \label{def:spiked-wishart}
    Let $\sX$ be a normalized spike prior, let $\gamma > 0$, and let $\beta \in [-1, \infty)$.
    Let $N = \lceil n/\gamma \rceil$. We define two probability distributions over $(\RR^n)^N$:
    \begin{enumerate}
    \item Under $\QQ$, the \emph{null} model, draw $\bm y_i \sim \sN(\bm 0, \bm I_n)$ independently for $i \in [N]$.
    \item Under $\PP$, the \emph{planted} model, draw $\bx \sim \sX_n$. If $\beta\|\bx\|^2 \geq -1$, then draw $\bm y_i \sim \sN(\bm 0, \bm I_n + \beta\bx\bx^\top)$ independently for $i \in [N]$. Otherwise, draw $\bm y_i \sim \sN(\bm 0, \bm I_n)$ independently for $i \in [N]$.
    \end{enumerate}
    Taken together, $\PP$ and $\QQ$ form the \emph{spiked Wishart model} $(\PP, \QQ) \equalscolon \Wishart(n, \gamma, \beta, \sX)$.
    For fixed $\gamma$ and $\beta$ we denote the sequence $(\Wishart(n, \gamma, \beta, \sX))_{n \in \NN}$ by $\Wishart(\gamma, \beta, \sX)$.
\end{definition}

\noindent
Several remarks on this definition are in order.
First, we make the explicit choice $N = \lceil n/\gamma \rceil$ for concreteness, but our results apply to any choice of $N = N(n)$ for which $n/N \to \gamma$ as $n \to \infty$.

Second, often the Wishart model is described in terms of the distribution of the sample covariance matrix $\frac{1}{N} \sum_{i=1}^N \by_i \by_i^\top$. We instead work directly with the samples $\by_i$ so as not to restrict ourselves to algorithms that only use the sample covariance matrix.
(This modification only makes our results on computational hardness of detection more general.)

Finally, the definition of $\PP$ has two cases to ensure that the covariance matrix $\bm I_n + \beta \bx \bx^\top$ is positive semidefinite. We will work in the setting $\beta > -1$ where the first case ($\beta \|\bx\|^2 \ge -1$) occurs with high probability. Priors for which this case occurs almost surely will be especially important, so we define the following terminology for this situation.
\begin{definition}
    Let $\beta \in (-1, \infty)$ and let $\sX$ be a normalized spike prior.
    We say that $\sX$ is \emph{$\beta$-good} if when $\bx \sim \sX_n$ then $\beta\|\bx\|^2 > -1$ almost surely.
\end{definition}

We will often consider spike priors having i.i.d.\ entries, and will sometimes need to slightly modify the spike prior to ensure that it is $\beta$-good and has bounded norm.
\begin{definition}\label{def:iid}
Let $\pi$ be a probability distribution over $\RR$ such that $\EE[\pi] = 0$ and $\EE[\pi^2] = 1$. Let $\iid(\pi/\sqrt{n})$ denote the normalized spike prior $\sX = (\sX_n)$ that draws each entry of $\bx$ independently from $\frac{1}{\sqrt n} \pi$. (We do not allow $\pi$ to depend on $n$.)
\end{definition}

\begin{definition}\label{def:trunc}
For a normalized spike prior $\sX$, let the $\beta$-truncation $\trunc_\beta(\sX)$ of $\sX$ denote the following normalized spike prior. To sample $\bx$ from $(\trunc_\beta(\sX))_n$, first sample $\bx' \sim \sX_n$. Then, let $\bx = \bx^\prime$ if $\beta\|\bx^\prime\|^2 > -1$ and $\|\bx^\prime\|^2 \leq 2$, and let $\bx = \bm 0$ otherwise.
\end{definition}

\noindent If $\beta > -1$ then since $\sX$ is normalized ($\|\bx'\| \to 1$ in probability), the first case of Definition~\ref{def:trunc} occurs with high probability. The upper bound $\|\bx'\| \le 2$ is for technical convenience, and the constant 2 is not essential. Note also that the $\beta$-truncation of an i.i.d.\ prior is no longer i.i.d.

We consider the algorithmic task of distinguishing between $\PP$ and $\QQ$ in the following sense.

\begin{definition}\label{def:detection}
For sequences of distributions $\PP = (\PP_n)_{n \in \NN}$ and $\QQ = (\QQ_n)_{n \in \NN}$ over measurable spaces $(\Omega_n, \sF_n)_{n \in \NN}$, an algorithm $f_n: \Omega_n \to \{0,1\}$ achieves \emph{strong detection} between $\PP$ and $\QQ$ if $$\QQ_n[f_n(\by) = 0] = 1-o(1) \qquad\text{and}\qquad \PP_n[f_n(\by) = 1] = 1-o(1).$$
\end{definition}
\noindent
The celebrated BBP transition \cite{BBP} implies a spectral algorithm for strong detection in the spiked Wishart model whenever $\beta^2 > \gamma$.
\begin{theorem}[\cite{BBP,BS-wishart}]
\label{thm:bbp}
Let $\sX$ be any normalized spike prior. If $\beta^2 > \gamma$ then there exists a polynomial-time algorithm for strong detection in $\Wishart(\gamma,\beta,\sX)$.
\end{theorem}
\noindent The algorithm computes the largest eigenvalue (if $\beta > 0$) or smallest eigenvalue (if $\beta < 0$) of the sample covariance matrix $\frac{1}{N} \sum_{i=1}^N \by_i \by_i^\top$. This eigenvalue converges almost surely to a limiting value which is different under $\PP$ and $\QQ$.

We will argue (see Corollary~\ref{cor:wish-hard}) that if $\sX = \iid(\pi / \sqrt{n})$ with $\pi$ subgaussian, then no polynomial-time algorithm achieves strong detection below the BBP threshold (when $\beta^2 < \gamma$). For some priors, exponential-time strong detection is possible below the BBP threshold \cite{PWBM-contig}. Very sparse priors with $\bx$ supported on $O(\sqrt{n})$ entries give rise to the \emph{sparse PCA} regime, where polynomial-time strong detection is possible below the BBP threshold \cite{JL-sparse,DM-sparse}. However, a normalized prior with this level of sparsity cannot take the form $\iid(\pi/\sqrt{n})$, since $\pi$ cannot depend on $n$.

\subsection{The Low-Degree Likelihood Ratio}\label{sec:low-deg}

Inspired by the sum-of-squares hierarchy (e.g., \cite{shor-sos,par-sos,las-sos}) and in particular the pseudo-calibration approach \cite{sos-clique}, recent works \cite{HS-bayesian,sos-power,sam-thesis} have proposed a strikingly simple method for predicting computational hardness of Bayesian inference problems. This method recovers widely-conjectured computational thresholds for high-dimensional inference problems such as planted clique, densest-$k$-subgraph, random constraint satisfaction, community detection, and sparse PCA (see \cite{sam-thesis}). We now give an overview of this method.

Consider the problem of distinguishing two simple hypotheses $\PP_n$ and $\QQ_n$ which are probability distributions on some domain $\Omega_n = \RR^{d(n)}$ (where typically the dimension $d(n)$ grows with $n$). One example is the spiked Wishart model $\Wishart(\gamma,\beta,\sX)$ for some fixed choice of the parameters $\beta,\gamma,\sX$. The idea is to take low-degree polynomials as a proxy for polynomial-time algorithms and consider whether there are low-degree polynomials $f_n: \Omega_n \to \RR$ that can distinguish $\PP_n$ from $\QQ_n$.

We view $\QQ_n$ as the ``null'' distribution, which is often i.i.d.\ Gaussian (as in the Wishart model) or i.i.d.\ Rademacher ($\pm 1$-valued).
$\QQ_n$ induces an inner product on $L^2$ functions $f: \Omega_n \to \RR$ given by $\la f,g \ra_{L^2(\QQ_{n})} = \EE_{\by \sim \QQ_n}[f(\by)g(\by)]$, and a norm $\|f\|_{L^2(\QQ_n)}^2 = \la f,f \ra_{L^2(\QQ_n)}$. For $D \in \NN$, let $\RR[\by]_{\le D}$ denote the polynomials $\Omega_n \to \RR$ of degree at most $D$. For $f: \Omega_n \to \RR$, let $f^{\le D}$ denote the orthogonal projection (with respect to $\la \cdot,\cdot \ra_{L^2(\QQ_n)}$) of $f$ onto $\RR[\by]_{\le D}$. The following relates the distinguishing power of low-degree polynomials to the \emph{low-degree likelihood ratio}.

\begin{theorem}[\cite{HS-bayesian}]
Let $\PP_n$ and $\QQ_n$ be probability distributions on $\Omega_n$ for each $n \in \NN$. Suppose $\PP_n$ is absolutely continuous with respect to $\QQ_n$, so that the likelihood ratio $L_n = \frac{d\PP_n}{d\QQ_n}$ is defined. Then
\begin{equation}\label{eq:max-f}
\max_{f \in \RR[\by]_{\le D} \setminus \{0\}} \frac{\EE_{\by \sim \PP_n} f(\by)}{\sqrt{\EE_{\by \sim \QQ_n} f(\by)^2}} = \|L_n^{\le D}\|_{L^2(\QQ_n)}.
\end{equation}
\end{theorem}
\begin{proof}
The left-hand side can be rewritten as
$$\max_{f \in \RR[\by]_{\le D} \setminus \{0\}} \frac{\la f,L_n \ra_{L^2(\QQ_n)}}{\|f\|_{L^2(\QQ_n)}},$$
so by basic Hilbert space theory, the maximum is attained by taking $f = L_n^{\le D}$.
\end{proof}

\noindent The left-hand side of (\ref{eq:max-f}) heuristically measures whether any degree-$D$ polynomial can distinguish $\PP_n$ from $\QQ_n$. Thus we expect $\|L_n^{\le D}\|_{L^2(\QQ_n)} = \omega(1)$ if some degree-$D$ polynomial that achieves strong detection, and $\|L_n^{\le D}\|_{L^2(\QQ_n)} = O(1)$ otherwise.

We take $O(\log n)$-degree polynomials $\Omega_n \to \RR$ as a proxy for functions computable in polynomial-time. One justification for this is that many polynomial-time algorithms compute the leading eigenvalue of a matrix $\bM$ whose entries are constant-degree polynomials in the data; in fact, there is formal evidence that such \emph{low-degree spectral methods} are as powerful as the sum-of-squares hierarchy \cite{sos-power}. Typically, $O(\log n)$ rounds of power iteration are sufficient to compute the leading eigenvalue accurately, which amounts to evaluating the $O(\log n)$-degree polynomial $\Tr(\bM^{2q})$ for some $q = O(\log n)$. This motivates the following informal conjecture.

\begin{conjecture}[Informal \cite{HS-bayesian,sos-power,sam-thesis}] \label{conj:low-degree}
For ``nice'' distributions $\PP_n$ and $\QQ_n$, if $\|L_n^{\le D}\|_{L^2(\QQ_n)} = O(1)$ for some $D = \log^{1+\Omega(1)}(n)$, then there is no randomized polynomial-time algorithm for strong detection between $\PP$ and $\QQ$.
\end{conjecture}

\noindent This conjecture is useful because the norm of the low-degree likelihood ratio, $\|L_n^{\le D}\|_{L^2(\QQ_n)}$, can be computed (or at least bounded) for various distributions such as the stochastic block model \cite{HS-bayesian} and the spiked tensor model \cite{sos-power,sam-thesis}.

\begin{remark}
We do not expect the converse of Conjecture~\ref{conj:low-degree} to hold. If $\|L_n^{\le D}\|_{L^2(\QQ_n)} = \omega(1)$ for some $D = O(\log n)$ then we expect an $n^{O(\log n)}$-time algorithm but not necessarily a polynomial-time algorithm, because not every $O(\log n)$-degree polynomial can be evaluated in polynomial time.
\end{remark}

Conjecture~\ref{conj:low-degree} is informal in that we do not specify what is meant by ``nice'' distributions. See \cite{sam-thesis} for a precise variant of Conjecture~\ref{conj:low-degree}; however, this variant uses the more refined notion of \emph{coordinate degree} and so does not apply to the calculations we will perform. Roughly speaking, ``nice'' distributions $\PP$ and $\QQ$ should satisfy the following:
\begin{enumerate}
    \item $\QQ$ should be a product distribution, e.g., i.i.d.\ Gaussian or i.i.d.\ Rademacher;
    \item $\PP$ should be sufficiently symmetric with respect to permutations of its coordinates; and
    \item we should be able to add a small amount of noise to $\PP$, ruling out distributions with brittle algebraic structure (such as random satisfiable instances of XOR-SAT, which can be identified using Gaussian elimination \cite{CD-xorsat}).
\end{enumerate}
We refer the reader to \cite{HS-bayesian,sam-thesis} for further details and evidence in favor of Conjecture~\ref{conj:low-degree}.

\section{Main Results}

\subsection{Spiked Wishart Models}

We expect that Conjecture~\ref{conj:low-degree} applies to the spiked Wishart model. The following states this assumption formally.

\begin{conjecture}\label{conj:wish-lowdeg}
Fix $\gamma > 0$, $\beta > -1$ and a normalized spike prior $\sX$. Let $\PP_n$ and $\QQ_n$ be the planted and null models, respectively, of $\Wishart(\gamma,\beta,\sX)$.
Define the likelihood ratio $L_n = L_{n,\gamma,\beta,\sX} = \frac{d\PP_{n}}{d\QQ_{n}}$. If there exists some $D = D(n) = \log^{1+\Omega(1)}(n)$ such that $\|L_{n}^{\le D}\|_{L^2(\QQ_n)} = O(1)$, then there is no randomized polynomial-time algorithm for strong detection in $\Wishart(\gamma,\beta,\sX)$.
\end{conjecture}

\noindent We now give bounds on $\|L_{n,\gamma,\beta,\sX}^{\le D}\|_{L^2(\QQ_n)}$.

\begin{theorem}\label{thm:wish-lowdeg}
Fix constants $\gamma > 0$ and $\beta > -1$.
\begin{enumerate}
    \item Suppose $\beta^2 < \gamma$. Let $\sX = \trunc_\beta(\iid(\pi/\sqrt{n}))$ where $\pi$ is subgaussian with $\EE[\pi] = 0$ and $\EE[\pi^2] = 1$. Then, for any $D = o(n/\log n)$, we have $\|L_{n, \gamma, \beta, \sX}^{\leq D}\|_{L^2(\QQ_{n})} = O(1)$.
    \item Suppose $\beta^2 > \gamma$. Let $\sX = \iid(\pi/\sqrt{n})$ be $\beta$-good with $\pi$ symmetric about zero, $\EE[\pi] = 0$, and $\EE[\pi^2] = 1$. Then, for any $D = \omega(1)$, we have $\|L_{n, \gamma, \beta, \sX}^{\leq D}\|_{L^2(\QQ_{n})} = \omega(1)$.
\end{enumerate}
\end{theorem}

\noindent We prove Theorem~\ref{thm:wish-lowdeg} in Section~\ref{sec:proof-wish}. Part 1 of Theorem~\ref{thm:wish-lowdeg}, combined with Conjecture~\ref{conj:wish-lowdeg}, implies that for i.i.d.\ subgaussian priors, strong detection is computationally hard below the BBP threshold.
\begin{corollary}\label{cor:wish-hard}
Suppose Conjecture~\ref{conj:wish-lowdeg} holds. Fix constants $\gamma > 0$ and $\beta > -1$. Let $\pi$ be subgaussian with $\EE[\pi]=0$ and $\EE[\pi^2]=1$. Let $\sX$ be either $\iid(\pi/\sqrt{n})$ or $\trunc_\beta(\iid(\pi/\sqrt{n}))$. If $\beta^2 < \gamma$, then there is no randomized polynomial-time algorithm for strong detection in $\Wishart(\gamma,\beta,\sX)$.
\end{corollary}
\begin{proof}
The case $\sX = \trunc_\beta(\iid(\pi/\sqrt{n}))$ follows immediately from Part~1 of Theorem~\ref{thm:wish-lowdeg}. If strong detection is impossible for $\sX = \trunc_\beta(\iid(\pi/\sqrt{n}))$, then strong detection is also impossible for $\sX = \iid(\pi/\sqrt{n})$, as these two spike priors differ with probability $o(1)$ (under the natural coupling).
\end{proof}
\noindent
A number of technical remarks on the content of Theorem~\ref{thm:wish-lowdeg} and Corollary~\ref{cor:wish-hard} are in order.

\begin{remark}
Even if Conjecture~\ref{conj:wish-lowdeg} does not hold, note that Theorem~\ref{thm:wish-lowdeg} still implies unconditional lower bounds against low-degree polynomials in the sense of \eqref{eq:max-f}.
\end{remark}

\begin{remark}
Part 2 of Theorem~\ref{thm:wish-lowdeg} serves only to check that we do not predict computational hardness when $\beta^2 > \gamma$ (as polynomial-time strong detection is possible in this regime; see Theorem~\ref{thm:bbp}). The assumption that $\pi$ is symmetric about zero should not be essential.
\end{remark}

\begin{remark}
In Part~1 of Theorem~\ref{thm:wish-lowdeg} and in Corollary~\ref{cor:wish-hard}, the requirement that $\sX$ be a $\beta$-truncated i.i.d.\ prior can be relaxed. We only require that $\sX$ it is the $\beta$-truncation of a normalized prior admitting a \emph{local Chernoff bound} (see Definition~\ref{def:local-chernoff}).
\end{remark}

\begin{remark}
Part~1 of Theorem~\ref{thm:wish-lowdeg} holds for any $D = o(n/\log n)$, much larger than the $D = \log^{1+\Omega(1)}(n)$ required by Conjecture~\ref{conj:wish-lowdeg}. Since degree-$D$ polynomials should correspond to $n^{\widetilde{ O}(D)}$-time algorithms \cite{sam-thesis}, this suggests that the conclusion of Corollary~\ref{cor:wish-hard} also holds for $2^{n^{1-\delta}}$-time algorithms, for any $\delta > 0$. In other words, strong detection requires nearly-exponential time.
\end{remark}


\subsection{Constrained PCA}

We now give a reduction from strong detection in the spiked Wishart model to certification in the constrained PCA problem.

\begin{theorem}\label{thm:reduction}
Let $\sS$ be a constraint set and let $\sX$ be a normalized spike prior such that if $\bx \sim \sX_n$ then $\bx \in \sS_n$ with high probability. Suppose there exists $\varepsilon > 0$ and a randomized polynomial-time algorithm that certifies the value $2-\varepsilon$ on $\PCA(\sS)$. Then there exist $\gamma > 1$ and $\beta \in (-1,0)$ such that there is a randomized polynomial-time algorithm for strong detection in $\Wishart(\gamma,\beta,\sX)$.
\end{theorem}
\noindent We give the proof in Section~\ref{sec:reduction}. Note for the parameters $\gamma,\beta$ above, $\beta^2 < \gamma$ (the ``hard regime'').

\begin{corollary}\label{cor:pca-hard}
Suppose Conjecture~\ref{conj:wish-lowdeg} holds. Let $\pi$ be subgaussian with $\EE[\pi] = 0$ and $\EE[\pi^2] = 1$. Let $\sS$ be a constraint set such that, if $\bx \sim \iid(\pi/\sqrt{n})$, then $\bx \in \sS_n$ with high probability. Then, for any $\varepsilon > 0$, there is no randomized polynomial-time algorithm to certify the value $2-\varepsilon$ on $\PCA(\sS)$.
\end{corollary}
\begin{proof}
The result is immediate from Theorem~\ref{thm:wish-lowdeg} and Theorem~\ref{thm:reduction}.
\end{proof}
\noindent In particular, we obtain the hardness of improving on the spectral certificate in the SK model.
\begin{corollary}
If Conjecture~\ref{conj:wish-lowdeg} holds, then, for any $\varepsilon > 0$, there is no randomized polynomial-time algorithm to certify the value $2-\varepsilon$ on the SK problem $\PCA(\{\pm 1/\sqrt{n}\}^n)$.
\end{corollary}
\begin{proof}
Apply Corollary~\ref{cor:pca-hard} with $\pi$ having the Rademacher distribution and $\sS_n = \{\pm 1/\sqrt{n}\}^n$.
\end{proof}

\section{Proof of Reduction from Spiked Wishart to Constrained PCA}
\label{sec:reduction}

\begin{proof}[Proof of Theorem~\ref{thm:reduction}]

Let $\sS$ be a constraint set and let $\sX$ be a normalized spike prior such that, if $\bx \sim \sX_n$, then $\bx \in \sS_n$ with high probability. Suppose that for some $\varepsilon > 0$ there is a randomized polynomial-time algorithm $f$ that certifies the value $2-\varepsilon$ on $\PCA(\sS)$. We will show that this implies that there is a polynomial-time algorithm for strong detection in $\Wishart(\gamma,\beta,\sX)$ with certain parameters $\gamma > 1$ and $\beta \in (-1, 0)$ (depending on $\varepsilon$). Note that these parameters lie in the ``hard'' regime $\beta^2 < \gamma$.

Our algorithm for detection in the Wishart model is as follows. Fix $\gamma > 1$, to be chosen later. Since $n/N \to \gamma$ we have $n > N$ (for sufficiently large $n$). Given samples $\by_1,\ldots,\by_N \sim \sN(\bm 0,\bm I_n+\beta \bx \bx^\top)$, let $V = \mathrm{span}\{\by_1,\ldots,\by_N\} \subseteq \RR^n$ and let $V^\perp$ be its orthogonal complement. We sample $\bW \in \RR^{n \times n}$ having the distribution $\GOE(n)$ conditioned on the event that the span of the top $n-N$ eigenvectors of $\bW$ is $V^\perp$. Concretely, we can obtain a sample in the following way. Let $\bv_1,\ldots,\bv_N$ be a uniformly random orthonormal basis for $V$ and let $\bv_{N+1},\ldots,\bv_n$ be a uniformly random orthonormal basis for $V^\perp$. Sample $\bW^\prime \sim \GOE(n)$ and let $\lambda_1 < \cdots < \lambda_n$ be the eigenvalues of $\bW^\prime$. Then, let $\bW \colonequals \sum_{i=1}^n \lambda_i \bv_i \bv_i^\top$.
Finally, run the certification algorithm $f$ for $\PCA(\sS)$ on $\bW$.
The detection algorithm $\widetilde{f}: (\RR^n)^N \to \{0, 1\}$ then thresholds $f(\bW)$:
\begin{equation}
    \widetilde{f}(\by_1, \dots, \by_N) = \left\{\begin{array}{ll} 0 \text{ (report samples came from null distribution } \QQ_n \text{)} & \text{if } f(\bW) \le 2-\varepsilon/2, \\ 1 \text{ (report samples came from planted distribution } \PP_n \text{)} & \text{if } f(\bW) > 2-\varepsilon/2. \end{array}\right.
\end{equation}

We now prove that $\widetilde{f}$ indeed achieves strong detection in $\Wishart(\gamma, \beta, \sX)$.
First, if the samples $\by_i$ are drawn from the null model $\QQ_n$, then $V$ is a uniformly random $N$-dimensional subspace of $\RR^n$, so by Proposition~\ref{prop:goe-invariance} the law of $\bW$ constructed above is $\GOE(n)$. Thus $f(\bW) \le 2 - \varepsilon/2$ with high probability by assumption, and therefore $\widetilde{f}(\by_1, \dots, \by_N) = 0$ with high probability, i.e., our algorithm correctly reports that the samples were drawn from the null model.

Next, suppose the samples $\by_i$ are drawn from the planted model $\PP_n$ with planted spike $\bx \sim \sX_n$. We will choose $\gamma > 1$ and $\beta \in (-1, 0)$ so that $\bx^\top \bW \bx \ge 2-\varepsilon/3$ with high probability. Since $\bx \in \sS_n$ with high probability, this would imply $f(\bW) \ge 2 - \varepsilon/3$, so we will have $\widetilde{f}(\by_1, \dots, \by_N) = 1$ with high probability, i.e., our algorithm will correctly report that the samples were drawn from the planted model.

It remains to show that $\bx^\top \bW \bx \ge 2-\varepsilon/3$. Let $\lambda_1 < \cdots < \lambda_n$ be the eigenvalues of $\bW$ and let $\bv_1,\ldots,\bv_n$ be the corresponding (unit-norm) eigenvectors. By Proposition~\ref{prop:goe-spectrum}, with high probability, for all $i \in [n]$, $\lambda_i \in [-2-o(1),2+o(1)]$. Furthermore, by the semicircle law \cite{wigner-semicircle}, with high probability, $\lambda_{N+1} \ge 2 - g(\gamma)$ where $g(\gamma) > 0$ is a function satisfying $g(\gamma) \to 0$ as $\gamma \to 1^+$ (recalling that $n/N \to \gamma$). Letting $\|\bx\|_V$ denote the norm of the orthogonal projection of $\bx$ onto $V$, we have, with high probability,
\begin{align}
\bx^\top \bW \bx &= \bx^\top \left(\sum_{i=1}^n \lambda_i \bv_i \bv_i^\top \right) \bx \nonumber \\
&= \sum_{i=1}^n \lambda_i \langle \bx,\bv_i \rangle^2\nonumber \\
&\ge \lambda_1 \|\bx\|^2_V + \lambda_{N+1} \|\bx\|^2_{V^\perp}\nonumber \\
&\ge (-2-o(1)) \|\bx\|^2_V + (2-g(\gamma)) (\|\bx\|^2-\|\bx\|^2_V)\nonumber \\
&= (2 - g(\gamma))\|\bx\|^2 + (-4 + g(\gamma) - o(1))\|\bx\|_V^2\nonumber \\
&\ge 2 - g(\gamma) - 4 \|\bx\|_V^2 - o(1).
\label{eq:xWx}
\end{align}
Thus we need to upper bound $\|\bx\|_V^2$. Let $\bP_V$ denote the orthogonal projection matrix onto $V$. Since $V$ is the span of $\{\by_1,\ldots,\by_N\}$, we have $\bP_V \preceq \frac{1}{\mu} \bY$ where
$$\bY = \frac{1}{N} \sum_{i=1}^N \by_i \by_i^\top$$
and $\mu$ is the smallest nonzero eigenvalue of $\bY$. (Here $\preceq$ denotes Loewner order.) Since $\bY$ is a spiked Wishart matrix, it follows from Theorem~1.2 of \cite{BS-wishart} that its smallest nonzero eigenvalue converges almost surely to $(1-\sqrt{\gamma})^2$ as $n \to \infty$. Thus we have $\mu = (1-\sqrt{\gamma})^2 + o(1)$. Therefore,
$$\|\bx\|_V^2 = \|\bP_V \bx\|^2 = \bx^\top \bP_V \bx \le \frac{1}{\mu} \bx^\top \bY \bx =\frac{1}{\mu N} \sum_{i=1}^N \langle \bx,\by_i \rangle^2.$$
We have $\by_i \sim \sN(\bm 0,\bm I_n + \beta \bx \bx^\top)$ and so $\langle \bx,\by_i \rangle \sim \sN(0,\bx^\top (\bm I_n + \beta \bx\bx^\top)\bx) = \sN(0,\|\bx\|^2 + \beta\|\bx\|^4)$.
Therefore, letting $a_N^2 = \sum_{i = 1}^N g_i^2$ for $g_i$ i.i.d.\ standard gaussian random variables, so that $a_N^2$ has the $\chi^2$ distribution with $N$ degrees of freedom, we have conditional on $\bx$ the distributional equality
\[ \bx^\top \bY \bx \stackrel{(d)}{=} (\|\bx\|^2 + \beta\|\bx\|^4)\frac{a_N^2}{N}. \]
Standard concentration inequalities imply $a_N^2 / N \in [1 - o(1), 1 + o(1)]$ with high probability, and therefore $\bx^\top \bY \bx = 1 + \beta + o(1)$ with high probability.
Thus, with high probability, we find
\begin{equation}\label{eq:xV}
\|\bx\|_V^2 = \frac{1+\beta}{(1-\sqrt{\gamma})^{2}} + o(1).
\end{equation}
Finally, we choose $\gamma > 1$ close enough to 1 so that $g(\gamma) \le \varepsilon/8$. By \eqref{eq:xV}, we can also choose $\beta \in (-1, 0)$ close enough to $-1$ so that $\|\bx\|_V^2 \le \varepsilon/32$ with high probability. Combining these, from \eqref{eq:xWx} it follows that $\bx^\top \bW \bx \ge 2 - \varepsilon/4 - o(1) \ge 2 - \varepsilon/3$ with high probability, completing the proof.
\end{proof}

\begin{remark}
We remark that we have ignored issues of numerical precision by assuming a model of computation where eigendecomposition computations can be done exactly in polynomial time. However, we believe all the operations we have used are stable, so that our reduction should also hold for weaker models of computation. (In particular, if we want to compute polynomially-many bits of precision of the $\PCA(\sS)$ instance, this should require only polynomially-many bits of precision in the eigendecomposition computation.)
\end{remark}

\section{Proofs for Spiked Wishart Models}
\label{sec:proof-wish}

\subsection{Preliminaries}

\paragraph{Spiked Wishart model statistics.}
The following formulae pertaining to the spiked Wishart model are derived in \cite{PWBM-contig}.
(Recall that in the spiked Wishart model, the parameter $N$ is determined by $n$ and $\gamma$ as $N = \lceil n/\gamma \rceil$.)
\begin{proposition}
    Suppose $\gamma > 0$, $\beta \in [-1, \infty)$, and $\sX$ is a $\beta$-good normalized spike prior.
    Then, the likelihood ratio of the null and planted probability distributions of Definition~\ref{def:spiked-wishart} is
    \begin{align}
    L_{n, \gamma, \beta, \sX}(\by_1, \dots, \by_N) &\colonequals \frac{d\PP_{n}}{d\QQ_{n}}(\by_1, \dots, \by_N) \nonumber \\ &= \Ex_{\bx \sim \sX_n}\left[\left(1 + \beta \|\bx\|^2\right)^{-N / 2}\prod_{i = 1}^N \exp\left(\frac{1}{2}\, \frac{\beta}{1 + \beta \|\bx\|^2} \la \bx, \by_i \ra^2 \right) \right].
    \end{align}
    If furthermore $\|\bx\|^2 < 1/|\beta|$ almost surely when $\bx \sim \sX_n$, then the second moment of the likelihood ratio is given by
    \begin{equation}\label{eq:2nd}
    \Ex_{\by \sim \QQ_{n}} \left(L_{n, \gamma, \beta, \sX}(\by_1, \dots, \by_N)\right)^2  = \Ex_{\bx^1, \bx^2 \sim \sX_n}\left[(1 - \beta^2\la \bx^1, \bx^2 \ra^2)^{-N / 2}\right]
    \end{equation}
    where $\bx^1,\bx^2$ are drawn independently from $\sX_n$.
\end{proposition}

\paragraph{Hermite polynomials.}

We recall the classical one-dimensional Hermite polynomials.
\begin{definition}
The polynomials $h_k \in \RR[x]$ for $k \geq 0$ are defined by the recursion
\begin{align*}
  h_0(x) &= 1, \\
  h_{k + 1}(x) &= xh_k(x) - h_k^{\prime}(x),
\end{align*}
and we define normalized versions
\begin{equation*}
    \what{h}_k(x) = \frac{1}{\sqrt{k!}} h_k(x).
\end{equation*}
\end{definition}

\begin{proposition}
The $\what{h}_k$ are an orthonormal polynomial system for the standard Gaussian measure:
\begin{equation*}
    \Ex_{g \sim \sN(0, 1)}\left[\what{h}_k(g)\what{h}_\ell(g)\right] = \delta_{k\ell}.
\end{equation*}
\end{proposition}

\noindent
Similarly, we define the product Hermite polynomials.
It is helpful to first define some notations for vectors of indices, which will also be used in the later derivations.
\begin{definition}
    Let $\NN = \{n \in \ZZ: n \geq 0\}$.
    For $\bm\alpha \in \NN^n$ and $\bx \in \RR^n$, let
    \begin{align*}
        |\bm\alpha| &\colonequals \sum_{i = 1}^n \alpha_i, \\
        \bm\alpha! &\colonequals \prod_{i = 1}^n \alpha_i!, \\
        \bx^{\bm\alpha} &\colonequals \prod_{i = 1}^n x_i^{\alpha_i}.
    \end{align*}
\end{definition}
\begin{definition}
    For $\bm\alpha \in \NN^n$ and $\bx \in \RR^n$,
\begin{align*}
  H_{\bm\alpha}(\bx) &\colonequals \prod_{i = 1}^n h_{\alpha_i}(x_i), \\
  \what{H}_{\bm\alpha}(\bx) &\colonequals \prod_{i = 1}^n \what{h}_{\alpha_i}(x_i) = \frac{1}{\sqrt{\bm\alpha!}} H_{\bm\alpha}(\bx).
\end{align*}
\end{definition}

\begin{proposition}
The $\what{H}_{\bm\alpha}$ are an orthonormal polynomial system for the product measure of $n$ standard Gaussian measures:
\[ \Ex_{\bg \sim \sN(\bm 0, \bm I_n)}\left[ \what{H}_{\bm\alpha}(\bg) \what{H}_{\bm\beta}(\bg)\right] = \delta_{\bm\alpha\bm\beta}. \]
\end{proposition}

\paragraph{Combinatorics.}

We will also need the (ordinary) generating function of the central binomial coefficients.
\begin{proposition}
    \label{prop:gf-central-binom}
    For any $x \in \RR$ with $|x| < \frac{1}{4}$,
    \[
    (1 - 4x)^{-1/2} = \sum_{k \geq 0} \binom{2k}{k} x^k. \]
\end{proposition}

\subsection{Norm of the Low-Degree Projection}

In this section, we describe the formulas for the norm of the low-degree likelihood ratio in the spiked Wishart model, $\|L_{n, \gamma, \beta, \sX}^{\leq D}\|_{L^2(\QQ_n)}$.
The full calculations are given in Appendix~\ref{app:orth-poly}.

The following result, the main technical one of this portion of the argument, computes the norm of the projection of the likelihood ratio onto a single Hermite polynomial.
\begin{lemma}
\label{lem:proj-component-norms}
Let $\bm\alpha \in (\NN^n)^N$, and let $\bm\alpha_i \in \NN^n$ denote the $i$th component.
Let $|\bm\alpha| = \sum_{i = 1}^N |\bm\alpha_i|$.
Suppose $\gamma > 0$, $\beta \in [-1, \infty)$, and $\sX$ is a $\beta$-good normalized spike prior.
Then,
\begin{align}
    &\la L_{n, \gamma, \beta, \sX}, \what{H}_{\bm\alpha}\ra_{L^2(\QQ_{n})}^2 \nonumber \\
    &\hspace{1cm}= \left\{\begin{array}{ll} \beta^{|\bm\alpha|} \cdot \prod_{i = 1}^N \frac{(|\bm\alpha_i| - 1)!!^2}{\bm\alpha_i!} \cdot \left( \EE_{\bx \sim \sX_n}\,\bx^{\sum_{i = 1}^N \bm\alpha_i}\right)^2 & \text{if } |\bm\alpha_i| \text{ even for all } i \in [N], \\ 0 & \text{otherwise,} \end{array} \right.
\end{align}
where when $\beta = 0$ and $\bm\alpha = \bm 0$ we interpret $0^0 = 1$.
\end{lemma}
\noindent
Note in particular that the quantity in question does not depend on the sign of $\beta$; thus the calculation of the norm of the low-degree projection of the likelihood ratio will not distinguish between the positively and negatively spiked Wishart models.
Interestingly, in our proof, which involves generalized Hermite polynomials that form families of orthogonal polynomials with respect to Gaussian measures of different variances, this corresponds to the fact that an ``umbral'' analogue of the Hermite polynomials corresponding to a fictitious Gaussian measure with \emph{negative} variance satisfies many of the same identities as the ordinary Hermite polynomials.

Combining these quantities, we may give a simple description of the norm of the low-degree projection of the likelihood ratio.
\begin{lemma}
    \label{lem:proj-norm}
    Suppose $\gamma > 0$, $\beta \in [-1, \infty)$, and $\sX$ is a $\beta$-good normalized spike prior.
    Define
    \begin{align}
        \varphi_N(x) &\colonequals (1 - 4x)^{-N / 2}, \\
        \varphi_{N, k}(x) &\colonequals \sum_{d=0}^{k} x^d \sum_{\substack{d_1, \dots, d_N \\ \sum d_i = d}} \prod_{i = 1}^N \binom{2d_i}{d_i}, \label{eq:taylor}
    \end{align}
    so that $\varphi_{N, k}(x)$ is the Taylor series of $\varphi_N$ around $x = 0$ truncated to degree $k$ (as may be justified by Proposition~\ref{prop:gf-central-binom}).
    Then,
    \begin{equation}\label{eq:LD}
        \|L_{n, \gamma, \beta, \sX}^{\leq D}\|_{L^2(\QQ_{n})}^2 = \Ex_{\bx^1, \bx^2 \sim \sX_n}\left[ \varphi_{N, \lfloor D / 2\rfloor}\left(\frac{\beta^2\la \bx^1, \bx^2 \ra^2}{4}\right)\right]
    \end{equation}
    where $\bx^1,\bx^2$ are drawn independently from $\sX$.
\end{lemma}

\begin{remark}
The squared norm of the low-degree likelihood ratio \eqref{eq:LD} is closely related via Taylor expansion to the squared norm (or second moment) of the full likelihood ratio \eqref{eq:2nd}, which is recovered by taking $D \to \infty$ while $n$ and $N$ remain fixed.
\end{remark}

\subsection{Asymptotics as $n \to \infty$}

In this section, we use the formula from Lemma~\ref{lem:proj-norm} to prove Part 1 of Theorem~\ref{thm:wish-lowdeg} (the case $\beta^2 < \gamma$). The proof of Part 2 ($\beta^2 > \gamma$) is deferred to Appendix~\ref{app:above-thresh}.

The following concentration result is the key property that we require from the spike prior $\sX$.

\begin{definition}\label{def:local-chernoff}
A normalized spike prior $\sX$ admits a \emph{local Chernoff bound} if for every $\eta > 0$ there exist $\delta > 0$ and $C > 0$ such that, for all $n$,
\begin{equation}\label{eq:local-chernoff}
\Pr\left\{|\la \bx^1,\bx^2 \ra| \ge t\right\} \le C \exp\left(-\frac{1}{2}(1-\eta)nt^2\right) \quad \text{for all } t \in [0,\delta]
\end{equation}
where $\bx^1,\bx^2$ are drawn independently from $\sX_n$.
\end{definition}

\begin{proposition}\label{prop:local-chernoff}
If $\pi$ is subgaussian with $\EE[\pi] = 0$ and $\EE[\pi^2] = 1$ then $\iid(\pi/\sqrt{n})$ and $\trunc_\beta(\iid(\pi/\sqrt{n}))$ (for any $\beta > -1$) each admit a local Chernoff bound.
\end{proposition}
\noindent We defer the proof to Appendix~\ref{app:local-chernoff}.

\begin{proof}[Proof of Theorem~\ref{thm:wish-lowdeg} (Part 1)]
Let $\beta^2 < \gamma$. We decompose the norm of the low-degree likelihood ratio into two parts, which we will bound separately:

\[ \|L_{n, \gamma, \beta, \sX}^{\leq D}\|_{L^2(\QQ_{n})}^2 = \Ex_{\bx^1, \bx^2 \sim \sX_n}\left[\varphi_{N, \lfloor D / 2\rfloor}\left(\frac{\beta^2}{4} \la \bx^1, \bx^2 \ra^2\right)\right] = R_1 + R_2 \]
where
\begin{align*}
R_1 &\colonequals \Ex_{\bx^1, \bx^2 \sim \sX_n}\left[\One_{|\langle \bx^1,\bx^2 \rangle| \le \varepsilon}\, \varphi_{N, \lfloor D / 2\rfloor}\left(\frac{\beta^2}{4} \la \bx^1, \bx^2 \ra^2\right)\right], \\
R_2 &\colonequals \Ex_{\bx^1, \bx^2 \sim \sX_n}\left[\One_{|\langle \bx^1,\bx^2 \rangle| > \varepsilon}\, \varphi_{N, \lfloor D / 2\rfloor}\left(\frac{\beta^2}{4} \la \bx^1, \bx^2 \ra^2\right)\right].
\end{align*}
Here $\varepsilon > 0$ is a small constant to be chosen later. We call $R_1$ the \emph{small deviations} and call $R_2$ the \emph{large deviations}.

The following two lemmas bound these two terms, respectively.
First, we bound the large deviations.
\begin{lemma}[Large Deviations]\label{lem:large-dev}
Let $\beta^2 < \gamma$. Suppose $\sX$ is a $\beta$-good normalized spike prior that admits a local Chernoff bound. Suppose that for any $n$, $\bx \sim \sX_n$ satisfies $\|\bx\|^2 \le 2$ almost surely. If $D = o(n/\log n)$ and $\varepsilon > 0$ is any constant, then $R_2 = o(1)$.
\end{lemma}
\noindent We give a proof summary, with the full proof deferred to Appendix~\ref{app:large-dev}. Since $\|\bx^1\|^2 \le 2$ and $\|\bx^2\|^2 \le 2$,
$$R_2 \le \Pr\left\{|\langle \bx^1,\bx^2 \rangle| > \varepsilon\right\} \varphi_{N, \lfloor D / 2\rfloor}(\beta^2).$$
By the local Chernoff bound, $\Pr\left\{|\langle \bx^1,\bx^2 \rangle| > \varepsilon\right\}$ decays exponentially in $n$. To complete the proof, we use elementary combinatorial bounds to control the polynomial expression \eqref{eq:taylor} for $\varphi_{N, \lfloor D / 2\rfloor}(\beta^2)$.
Its growth is roughly of order $O(n^D)$, which is counteracted by the exponential decay of $\Pr\left\{|\langle \bx^1,\bx^2 \rangle| > \varepsilon\right\}$ so long as $D = o(n / \log n)$.

Next, we bound the small deviations.
For this part of the argument, it is irrelevant that the likelihood ratio is truncated to its low-degree component, and we essentially reuse an existing argument for the full likelihood ratio from \cite{PWBM-contig}.
\begin{lemma}[Small Deviations]\label{lem:small-dev}
Let $\beta^2 < \gamma$. Suppose $\sX$ is a $\beta$-good normalized spike prior that admits a local Chernoff bound. Let $D = D(n)$ be any function of $n$. If $\varepsilon > 0$ is a sufficiently small constant then $R_1 = O(1)$.
\end{lemma}
\noindent We again give a proof summary, with the full proof deferred to Appendix~\ref{app:small-dev}. As mentioned above, unlike in the proof of Lemma~\ref{lem:large-dev}, here we simply bound $\varphi_{N, \lfloor D / 2\rfloor} \leq \varphi_N$ in the expression for $R_1$. To bound the resulting expression, we borrow an argument from \cite{PWBM-contig}. This step crucially uses the local Chernoff bound, and amounts to showing that the exponential decay from the Chernoff bound sufficiently counteracts the exponential growth of the likelihood ratio term $\varphi_N(\frac{\beta^2}{4}\la \bx^1, \bx^2 \ra^2)$ when $\la \bx^1, \bx^2 \ra$ is small.

Combining Proposition~\ref{prop:local-chernoff} with Lemmas~\ref{lem:large-dev} and \ref{lem:small-dev} completes the proof of Part 1 of Theorem~\ref{thm:wish-lowdeg}. (The proof of Part 2 is deferred to Appendix~\ref{app:above-thresh}.)
\end{proof}

\section*{Acknowledgements}
We thank Andrea Montanari and Samuel B.\ Hopkins for insightful discussions.

\bibliographystyle{alpha}
\bibliography{main}

\newcommand{\etalchar}[1]{$^{#1}$}
\begin{thebibliography}{PWBM18b}

\bibitem[ABM18]{ABM-crem}
Louigi Addario-Berry and Pascal Maillard.
\newblock The algorithmic hardness threshold for continuous random energy
  models.
\newblock {\em arXiv preprint arXiv:1810.05129}, 2018.

\bibitem[AC08]{coja-barriers}
Dimitris Achlioptas and Amin {Coja-Oghlan}.
\newblock Algorithmic barriers from phase transitions.
\newblock In {\em Foundations of Computer Science, 2008. FOCS'08. IEEE 49th
  Annual IEEE Symposium on}, pages 793--802. IEEE, 2008.

\bibitem[AGZ10]{AGZ-book}
Greg~W Anderson, Alice Guionnet, and Ofer Zeitouni.
\newblock An introduction to random matrices, 2010.

\bibitem[ALR87]{ALR-sk}
Michael Aizenman, Joel~L Lebowitz, and David Ruelle.
\newblock Some rigorous results on the {Sherrington-Kirkpatrick} spin glass
  model.
\newblock {\em Communications in mathematical physics}, 112(1):3--20, 1987.

\bibitem[BBLS18]{mra-het}
Nicolas Boumal, Tamir Bendory, Roy~R Lederman, and Amit Singer.
\newblock Heterogeneous multireference alignment: A single pass approach.
\newblock In {\em Information Sciences and Systems (CISS), 2018 52nd Annual
  Conference on}, pages 1--6. IEEE, 2018.

\bibitem[BBP05]{BBP}
Jinho Baik, G{\'e}rard {Ben Arous}, and Sandrine P{\'e}ch{\'e}.
\newblock Phase transition of the largest eigenvalue for nonnull complex sample
  covariance matrices.
\newblock {\em The Annals of Probability}, 33(5):1643--1697, 2005.

\bibitem[BHK{\etalchar{+}}16]{sos-clique}
Boaz Barak, Samuel~B Hopkins, Jonathan Kelner, Pravesh Kothari, Ankur Moitra,
  and Aaron Potechin.
\newblock A nearly tight sum-of-squares lower bound for the planted clique
  problem.
\newblock In {\em Foundations of Computer Science (FOCS), 2016 IEEE 57th Annual
  Symposium on}, pages 428--437. IEEE, 2016.

\bibitem[BKM17]{BKM-coloring}
Jess Banks, Robert Kleinberg, and Cristopher Moore.
\newblock The lov$\backslash$'asz theta function for random regular graphs and
  community detection in the hard regime.
\newblock {\em arXiv preprint arXiv:1705.01194}, 2017.

\bibitem[BS06]{BS-wishart}
Jinho Baik and Jack~W Silverstein.
\newblock Eigenvalues of large sample covariance matrices of spiked population
  models.
\newblock {\em Journal of multivariate analysis}, 97(6):1382--1408, 2006.

\bibitem[CD99]{CD-xorsat}
Nadia Creignou and Herv{\'e} Daude.
\newblock Satisfiability threshold for random {XOR-CNF} formulas.
\newblock {\em Discrete Applied Mathematics}, 96:41--53, 1999.

\bibitem[CGL04]{coja-refutation}
Amin {Coja-Oghlan}, Andreas Goerdt, and Andr{\'e} Lanka.
\newblock Strong refutation heuristics for random k-{SAT}.
\newblock In {\em Approximation, Randomization, and Combinatorial Optimization.
  Algorithms and Techniques}, pages 310--321. Springer, 2004.

\bibitem[CHS93]{pspin-2}
A~Crisanti, H~Horner, and H-J Sommers.
\newblock The spherical p-spin interaction spin-glass model.
\newblock {\em Zeitschrift f{\"u}r Physik B Condensed Matter}, 92(2):257--271,
  1993.

\bibitem[CR02]{cris-sk}
Andrea Crisanti and Tommaso Rizzo.
\newblock Analysis of the infinity-replica symmetry breaking solution of the
  {Sherrington-Kirkpatrick} model.
\newblock {\em Physical Review E}, 65(4):046137, 2002.

\bibitem[CS92]{pspin-1}
A~Crisanti and HJ~Sommers.
\newblock The spherical p-spin interaction spin glass model: the statics.
\newblock {\em Phys. B}, 87:341, 1992.

\bibitem[DKMZ11]{block-model}
Aurelien Decelle, Florent Krzakala, Cristopher Moore, and Lenka Zdeborov{\'a}.
\newblock Asymptotic analysis of the stochastic block model for modular
  networks and its algorithmic applications.
\newblock {\em Physical Review E}, 84(6):066106, 2011.

\bibitem[DM14a]{info-sparse-pca}
Yash Deshpande and Andrea Montanari.
\newblock Information-theoretically optimal sparse {PCA}.
\newblock {\em arXiv preprint arXiv:1402.2238}, 2014.

\bibitem[DM14b]{DM-sparse}
Yash Deshpande and Andrea Montanari.
\newblock Sparse {PCA} via covariance thresholding.
\newblock In {\em Advances in Neural Information Processing Systems}, pages
  334--342, 2014.

\bibitem[DM15]{DM-clique}
Yash Deshpande and Andrea Montanari.
\newblock Improved sum-of-squares lower bounds for hidden clique and hidden
  submatrix problems.
\newblock In {\em Conference on Learning Theory}, pages 523--562, 2015.

\bibitem[FP07]{fp}
Delphine F{\'e}ral and Sandrine P{\'e}ch{\'e}.
\newblock The largest eigenvalue of rank one deformation of large {Wigner}
  matrices.
\newblock {\em Communications in mathematical physics}, 272(1):185--228, 2007.

\bibitem[GL18]{gamarnik-submatrix}
David Gamarnik and Quan Li.
\newblock Finding a large submatrix of a gaussian random matrix.
\newblock {\em The Annals of Statistics}, 46(6A):2511--2561, 2018.

\bibitem[HKP{\etalchar{+}}17]{sos-power}
Samuel~B Hopkins, Pravesh~K Kothari, Aaron Potechin, Prasad Raghavendra, Tselil
  Schramm, and David Steurer.
\newblock The power of sum-of-squares for detecting hidden structures.
\newblock In {\em Foundations of Computer Science (FOCS), 2017 IEEE 58th Annual
  Symposium on}, pages 720--731. IEEE, 2017.

\bibitem[Hop18]{sam-thesis}
Samuel Hopkins.
\newblock {\em Statistical Inference and the Sum of Squares Method}.
\newblock PhD thesis, Cornell University, 2018.

\bibitem[HS17]{HS-bayesian}
Samuel~B Hopkins and David Steurer.
\newblock Bayesian estimation from few samples: community detection and related
  problems.
\newblock {\em arXiv preprint arXiv:1710.00264}, 2017.

\bibitem[HSS15]{tensor-pca}
Samuel~B Hopkins, Jonathan Shi, and David Steurer.
\newblock Tensor principal component analysis via sum-of-squares proofs.
\newblock In {\em Conference on Learning Theory}, pages 956--1006, 2015.

\bibitem[Jer92]{jerrum-clique}
Mark Jerrum.
\newblock Large cliques elude the {Metropolis} process.
\newblock {\em Random Structures \& Algorithms}, 3(4):347--359, 1992.

\bibitem[JKR18]{JKR-meanfield}
Vishesh Jain, Frederic Koehler, and Andrej Risteski.
\newblock Mean-field approximation, convex hierarchies, and the optimality of
  correlation rounding: a unified perspective.
\newblock {\em arXiv preprint arXiv:1808.07226}, 2018.

\bibitem[JL04]{JL-sparse}
Iain~M Johnstone and Arthur~Yu Lu.
\newblock Sparse principal components analysis.
\newblock {\em Unpublished manuscript}, 7, 2004.

\bibitem[Joh01]{J-spike}
Iain~M. Johnstone.
\newblock On the distribution of the largest eigenvalue in principal components
  analysis.
\newblock {\em The Annals of Statistics}, 29(2):295--327, 2001.

\bibitem[KBG17]{hypergraph}
Chiheon Kim, Afonso~S Bandeira, and Michel~X Goemans.
\newblock Community detection in hypergraphs, spiked tensor models, and
  sum-of-squares.
\newblock In {\em Sampling Theory and Applications (SampTA), 2017 International
  Conference on}, pages 124--128. IEEE, 2017.

\bibitem[KKMO07]{KKMO-csp}
Subhash Khot, Guy Kindler, Elchanan Mossel, and Ryan O’Donnell.
\newblock Optimal inapproximability results for max-cut and other 2-variable
  {CSPs}?
\newblock {\em SIAM Journal on Computing}, 37(1):319--357, 2007.

\bibitem[KMOW17]{sos-refute}
Pravesh~K Kothari, Ryuhei Mori, Ryan O'Donnell, and David Witmer.
\newblock Sum of squares lower bounds for refuting any {CSP}.
\newblock In {\em Proceedings of the 49th Annual ACM SIGACT Symposium on Theory
  of Computing}, pages 132--145. ACM, 2017.

\bibitem[KXZ16]{mi-rank-1}
Florent Krzakala, Jiaming Xu, and Lenka Zdeborov{\'a}.
\newblock Mutual information in rank-one matrix estimation.
\newblock {\em arXiv preprint arXiv:1603.08447}, 2016.

\bibitem[Las01]{las-sos}
Jean~B Lasserre.
\newblock An explicit exact {SDP} relaxation for nonlinear 0-1 programs.
\newblock In {\em International Conference on Integer Programming and
  Combinatorial Optimization}, pages 293--303. Springer, 2001.

\bibitem[LKZ15]{mmse-rank-1}
Thibault Lesieur, Florent Krzakala, and Lenka Zdeborov{\'a}.
\newblock {MMSE} of probabilistic low-rank matrix estimation: Universality with
  respect to the output channel.
\newblock {\em arXiv preprint arXiv:1507.03857}, 2015.

\bibitem[Mon18]{mon-sk}
Andrea Montanari.
\newblock Optimization of the {Sherrington-Kirkpatrick Hamiltonian}.
\newblock {\em arXiv preprint arXiv:1812.10897}, 2018.

\bibitem[MPW15]{MPW-clique}
Raghu Meka, Aaron Potechin, and Avi Wigderson.
\newblock Sum-of-squares lower bounds for planted clique.
\newblock In {\em Proceedings of the forty-seventh annual ACM symposium on
  Theory of computing}, pages 87--96. ACM, 2015.

\bibitem[MR16]{positive-pca}
Andrea Montanari and Emile Richard.
\newblock Non-negative principal component analysis: Message passing algorithms
  and sharp asymptotics.
\newblock {\em IEEE Transactions on Information Theory}, 62(3):1458--1484,
  2016.

\bibitem[MS15]{MS-sdp}
Andrea Montanari and Subhabrata Sen.
\newblock Semidefinite programs on sparse random graphs and their application
  to community detection.
\newblock {\em arXiv preprint arXiv:1504.05910}, 2015.

\bibitem[Pan13]{pan-sk}
Dmitry Panchenko.
\newblock {\em The {Sherrington-Kirkpatrick} model}.
\newblock Springer Science \& Business Media, 2013.

\bibitem[Par79]{par-frsb}
Giorgio Parisi.
\newblock Infinite number of order parameters for spin-glasses.
\newblock {\em Physical Review Letters}, 43(23):1754, 1979.

\bibitem[Par80]{par-sk-num}
Giorgio Parisi.
\newblock A sequence of approximated solutions to the {SK} model for spin
  glasses.
\newblock {\em Journal of Physics A: Mathematical and General}, 13(4):L115,
  1980.

\bibitem[Par00]{par-sos}
Pablo~A Parrilo.
\newblock {\em Structured semidefinite programs and semialgebraic geometry
  methods in robustness and optimization}.
\newblock PhD thesis, California Institute of Technology, 2000.

\bibitem[PWBM18a]{synch-amp}
Amelia Perry, Alexander~S Wein, Afonso~S Bandeira, and Ankur Moitra.
\newblock Message-passing algorithms for synchronization problems over compact
  groups.
\newblock {\em Communications on Pure and Applied Mathematics},
  71(11):2275--2322, 2018.

\bibitem[PWBM18b]{PWBM-contig}
Amelia Perry, Alexander~S Wein, Afonso~S Bandeira, and Ankur Moitra.
\newblock Optimality and sub-optimality of {PCA I}: Spiked random matrix
  models.
\newblock {\em The Annals of Statistics}, 46(5):2416--2451, 2018.

\bibitem[RH18]{rig-notes}
Philippe Rigollet and Jan-Christian H\"utter.
\newblock High-dimensional statistics.
\newblock {\em Lecture notes}, 2018.

\bibitem[Rom05]{R-umbral}
Steven Roman.
\newblock {\em The umbral calculus}.
\newblock Springer, 2005.

\bibitem[Sho87]{shor-sos}
Naum~Z Shor.
\newblock Class of global minimum bounds of polynomial functions.
\newblock {\em Cybernetics}, 23(6):731--734, 1987.

\bibitem[SK75]{sk}
David Sherrington and Scott Kirkpatrick.
\newblock Solvable model of a spin-glass.
\newblock {\em Physical review letters}, 35(26):1792, 1975.

\bibitem[Sub18]{sub-frsb}
Eliran Subag.
\newblock Following the ground-states of full-{RSB} spherical spin glasses.
\newblock {\em arXiv preprint arXiv:1812.04588}, 2018.

\bibitem[Tal06]{tal-parisi}
Michel Talagrand.
\newblock The {Parisi} formula.
\newblock {\em Annals of mathematics}, pages 221--263, 2006.

\bibitem[Wig93]{wigner-semicircle}
Eugene~P Wigner.
\newblock Characteristic vectors of bordered matrices with infinite dimensions
  {I}.
\newblock In {\em The Collected Works of Eugene Paul Wigner}, pages 524--540.
  Springer, 1993.

\bibitem[ZK08]{quiet-1}
Lenka Zdeborova and Florent Krzakala.
\newblock Hiding quiet solutions in random constraint satisfaction problems.
\newblock {\em Physical Review Letters}, 102(LA-UR-08-08090; LA-UR-08-8090),
  2008.

\bibitem[ZK11]{quiet-2}
Lenka Zdeborov{\'a} and Florent Krzakala.
\newblock Quiet planting in the locked constraint satisfaction problems.
\newblock {\em SIAM Journal on Discrete Mathematics}, 25(2):750--770, 2011.

\end{thebibliography}

\newpage

\appendix

\section{Proofs for Computing the Low-Degree Likelihood Ratio}
\label{app:orth-poly}

\subsection{Generalized and Umbral Hermite Polynomials}

We introduce some calculations with a useful generalization of the Hermite polynomials.
While the usual Hermite polynomials are a family of orthogonal polynomials for the Gaussian measure with variance 1, and a straightforward generalization yields orthogonal polynomials for the Gaussian measure with any positive variance, we will use the surprising further generalization to fictitious Gaussian measures with \emph{negative} variance, as described by the so-called \emph{umbral calculus}.
We follow the presentation of \cite{R-umbral} (specifically, Section 2.1 of Chapter 4).
\begin{definition}
    For any $v \in \RR$, the \emph{Hermite polynomials with variance $v$} are defined by the recursion
    \begin{align}
        h_0(x; v) &= 1, \\
        h_{k + 1}(x; v) &= xh_k(x; v) - v \partial_x[h_{k}](x; v).
    \end{align}
\end{definition}
\noindent
The next facts are useful for translating between different versions of the basic recursion and other properties of the Hermite polynomials.
\begin{proposition}[Differentiation Identity]
    For any $v, x \in \RR$,
    \begin{equation}
        \partial_x [h_k](x; v) = kh_{k - 1}(x; v).
    \end{equation}
\end{proposition}
\begin{proposition}[Alternate Recursion]
    For any $v \in \RR$, the Hermite polynomials are equivalently defined by the recursion
    \begin{align}
        h_0(x; v) &= 1, \\
        h_{k + 1}(x; v) &= xh_k(x; v) - v k h_{k - 1}(x; v).
    \end{align}
\end{proposition}
\noindent
The following is yet another common way of defining the Hermite polynomials, in terms of the derivatives of the corresponding Gaussian density (or, in the negative variance case, a suitable generalization thereof).
\begin{proposition}[Rodrigues Formula]
    \label{prop:rodrigues}
    Let $v \in \RR$ with $v \neq 0$.
    Then,
    \begin{equation}
    \frac{d^k}{dx^k}\left[\exp\left(-\frac{1}{2v}x^2\right)\right] = (-v)^{-k}h_k(x; v)\exp\left(-\frac{1}{2v}x^2\right).
    \end{equation}
\end{proposition}
\noindent
The next fact shows how the generalized Hermite polynomials transform under scaling.
\begin{proposition}[Scaling Identity]
    Let $v, w, x \in \RR$, then
    \begin{equation}
        h_k(wx; v) = w^k h_k\left(x; \frac{v}{w^2}\right).
    \end{equation}
\end{proposition}
\noindent
Finally, the following is a generalized version of Gaussian integration by parts.
We only provide the version of this identity for the standard Hermite polynomials, which is the only one we will use, but analogous statements hold for the generalized and umbral Hermite polynomials.
\begin{proposition}[Integration by Parts]
    \label{prop:hermite-ibp}
    Let $f \in \sC^k(\RR)$ have $|f^{(i)}(x)| \leq e^{Cx}$ for all $i \in \{0, 1, \dots, k\}$ and some $C > 0$.
    Then,
    \begin{equation}
        \Ex_{g \sim \sN(0, 1)}\left[ h_k(g; 1)f(g) \right] = \Ex_{g \sim \sN(0, 1)}\left[ f^{(k)}(g) \right].
    \end{equation}
\end{proposition}

While the above results are standard, we now give two results we will use in our calculation that do not seem to appear explicitly in the previous literature, although they are straightforward to obtain from the preceding facts.
First, we will use the following slightly more general version of the Rodgrigues formula (Proposition~\ref{prop:rodrigues}) in our calculations.
\begin{proposition}[Multidimensional Rodrigues Formula]
    \label{prop:ip-gaussian-diff}
    Let $\bx \in \RR^n$, $\bm\alpha \in \NN^n$, and $v \in \RR$ with $v \neq 0$.
    Then,
    \begin{equation}
        \partial_{\bm y}^{\bm\alpha}\left[\exp\left(-\frac{1}{2v}\la \bx, \by \ra^2\right)\right] = (-v)^{-|\bm\alpha|}\bx^{\bm\alpha} h_{|\bm\alpha|}(\la \bx, \by \ra; v)\exp\left(-\frac{1}{2v}\la \bx, \by \ra^2\right).
    \end{equation}
\end{proposition}
\begin{proof}
    We proceed by induction on $|\bm\alpha|$.
    Clearly the result holds for $\bm\alpha = \bm 0$.
    Suppose the result holds for all $|\bm\alpha^\prime| \leq k$, and $|\bm\alpha| = k + 1 > 0$.
    Let $\bm\alpha^\prime$ having $|\bm\alpha^\prime| = k$ differ from $\bm\alpha$ only in coordinate $i$, so that $\alpha_i^\prime = \alpha_i - 1$ and $\alpha_j^\prime = \alpha_j$ for all $j \neq i$.
    Then,
    \begin{align*}
    &\partial_{\bm y}^{\bm\alpha}\left[\exp\left(-\frac{1}{2v}\la \bx, \by \ra^2\right)\right] \nonumber \\
    &\hspace{2cm}= \partial_{y_i}\left[\partial_{\bm y}^{\bm\alpha^\prime}\left[\exp\left(-\frac{1}{2v}\la \bx, \by \ra^2\right)\right]\right] \\
    &\hspace{2cm}= (-v)^{-k}\bx^{\bm\alpha^\prime}\partial_{y_i}\left[h_{k}(\la \bx, \by \ra; v)\exp\left(-\frac{1}{2v}\la \bx, \by \ra^2\right)\right] \\
    &\hspace{2cm}= (-v)^{-k}\bx^{\bm\alpha^\prime}\bigg(x_i \partial_x[h_{k}](\la \bx, \by \ra; v) - v^{-1}\la \bx, \by \ra x_i h_k(\la \bx, \by \ra; v)\bigg)\exp\left(-\frac{1}{2v}\la \bx, \by \ra^2\right) \\
    &\hspace{2cm}= (-v)^{-(k + 1)}\bx^{\bm\alpha}\bigg( \la \bx, \by \ra h_k(\la \bx, \by \ra; v) - v\partial_x[h_{k}](\la \bx, \by \ra; v)\bigg)\exp\left(-\frac{1}{2v}\la \bx, \by \ra^2\right) \\
    &\hspace{2cm}= (-v)^{-(k + 1)}\bx^{\bm\alpha}h_{k + 1}(\la \bx, \by \ra; v)\exp\left(-\frac{1}{2v}\la \bx, \by \ra^2\right),
    \end{align*}
    completing the proof.
\end{proof}

\noindent
Second, we will need the following calculation evaluating the expectation of any Hermite polynomial under any centered Gaussian measure.
\begin{proposition}[Expectation Under Mismatched Variance]
    \label{prop:hermite-mismatched}
    Let $v \in \RR$ and $k \geq 0$.
    Then,
    \[ \EE_{g \sim \sN(0, \sigma^2)}\left[h_k(g; v)\right] = \left\{ \begin{array}{ll} 0 & \text{if } k \text{ odd,} \\ (k - 1)!! (\sigma^2 - v)^{k/2} & \text{if } k \text{ even.}\end{array}\right. \]
\end{proposition}
\begin{proof}
    The result for odd $k$ holds since $h_k(
    \cdot \hspace{0.05cm}; v)$ is an odd function for any $v \in \RR$ in this case.
    For even $k$, we argue by induction on $k$.
    The result clearly holds for $k = 0$.
    If the result holds for a given $k$, then we may compute
    \begin{align*}
      \EE_{g \sim \sN(0, \sigma^2)}[h_{k + 2}(g; v)]
      &= \EE_{g \sim \sN(0, 1)}[h_{k + 2}(\sigma g)] \\
      &= \sigma \EE_{g \sim \sN(0, 1)}[gh_{k + 1}(\sigma g; v)] - v(k + 1) \EE_{g \sim \sN(0, 1)}[h_k(\sigma g)] \\
      &= \sigma^2 \EE_{g \sim \sN(0, 1)}[\partial_x[h_{k +1}](\sigma g; v)] - v(k + 1)\EE_{g \sim \sN(0, 1)}[h_k(\sigma g)] \\
      &= (k + 1)\sigma^2 \EE_{g \sim \sN(0, 1)}[h_{k}(\sigma g; v)] - v(k + 1)\EE_{g \sim \sN(0, 1)}[h_k(\sigma g)] \\
      &= (k + 1)(\sigma^2 - v)\EE_{g \sim \sN(0, 1)}[h_k(\sigma g)] \\
      &= (k + 1)(\sigma^2 - v)\EE_{g \sim \sN(0, \sigma^2)}[h_k(g)]
    \end{align*}
    completing the proof.
\end{proof}
\noindent
We note two interesting features of this result.
First, it generalizes two simple cases, on the one hand $v = \sigma^2$ where the expectation is zero unless $k = 0$, as may be seen from the orthogonality relations, and on the other $v = 0$ where it recovers the moments of a Gaussian measure.
Second, the quantities appearing on the right-hand side formally resemble the moments of a Gaussian measure of suitable variance, but the formula in fact still holds for $\sigma^2 < v$, in which case case these quantities may be viewed as the moments of a fictitious Gaussian measure of negative variance (the same as inspired the umbral Hermite polynomials).

\subsection{Individual Hermite Components of the Likelihood Ratio}

\begin{proof}[Proof of Lemma~\ref{lem:proj-component-norms}]
By Proposition~\ref{prop:hermite-ibp}, we find
\begin{align}
  &\la L_{n, \gamma, \beta, \sX}, \what{H}_{\bm\alpha} \ra^2\nonumber \\
  &\hspace{0.5cm} = \frac{1}{\bm\alpha!} \left(\Ex_{\bm y \sim \QQ_{n}}\partial^{\bm\alpha}_{\by} L_{n, \gamma, \beta, \sX}(\by_1, \dots, \by_N)\right)^2 \nonumber \\
  &\hspace{0.5cm} = \frac{1}{\bm\alpha!} \left(\Ex_{\bm x \sim \sX_n, \bm y \sim \QQ_{n}}\left(1 + \beta \|\bx\|^2\right)^{-N / 2}\partial^{\bm\alpha}_{\by}\prod_{i = 1}^N \exp\left(\frac{1}{2} \frac{\beta}{1 + \beta \|\bx\|^2} \la \bx, \by_i \ra^2 \right) \right)^2 \nonumber \\
  &\hspace{0.5cm} =  \frac{1}{\prod_{i = 1}^N\bm\alpha_i!} \left(\Ex_{\bm x \sim \sX_n}\left(1 + \beta \|\bx\|^2\right)^{-N / 2}\prod_{i = 1}^N \Ex_{\bm y \sim \sN(\bm 0, \bm I_n)}\left[\partial^{\bm\alpha_i}_{\bm y} \exp\left(\frac{1}{2} \frac{\beta}{1 + \beta \|\bx\|^2} \la \bx, \by \ra^2 \right)\right] \right)^2,
\end{align}
where the $\bm\alpha_i \in \NN^n$ are the components of $\bm\alpha$ corresponding to $\by_i$, for each $i \in [N]$.
When $\beta = 0$, our result follows from the above, giving $\la L_{n, \gamma, \beta, \sX}, \what{H}_{\bm\alpha}\ra^2 = \delta_{0, |\bm\alpha|}$.
(Indeed, in this case the null and planted models are identical, so $L_{n, \gamma, \beta, \sX} = 1$ is a constant, which is compatible with the above.)
Let us suppose $\beta \neq 0$ below.

In this case, using Proposition~\ref{prop:ip-gaussian-diff}, we have
\begin{align}
    &\partial^{\bm\alpha_i}_{\bm y} \exp\left(\frac{1}{2} \frac{\beta}{1 + \beta \|\bx\|^2} \la \bx, \by \ra^2 \right) \nonumber \\
    &\hspace{1cm}= \left(\frac{1 + \beta \|\bx\|^2}{\beta}\right)^{-|\bm\alpha_i|}\bx^{\bm\alpha_i}h_{|\bm\alpha_i|}\left(\la \bx, \by \ra; -\frac{1 + \beta \|\bx\|^2}{\beta}\right)\exp\left(\frac{1}{2} \frac{\beta}{1 + \beta \|\bx\|^2} \la \bx, \by \ra^2 \right).
\end{align}
(Note that the sign of the spike, or equivalently the sign of $\beta$, is the opposite of the sign of the variance of the Hermite polynomials that appear; thus, it is the negatively spiked case that corresponds to the more natural positive variance Hermite polynomials.)
Since when $\by \sim \sN(\bm 0, \bm I_n)$ then $\la \bx, \by \ra \sim \sN(0, \|\bx\|^2)$, we find
\begin{align}
    \la L_{n, \gamma, \beta, \sX}, \what{H}_{\bm\alpha} \ra^2 = &\frac{\beta^{2|\bm\alpha|}}{\prod_{i = 1}^N\bm\alpha_i!} \Bigg(\Ex_{\bm x \sim \sX_n}\frac{\bx^{\sum_{i = 1}^N\bm\alpha_i}}{\left(1 + \beta \|\bx\|^2\right)^{|\bm\alpha| + N / 2}} \nonumber \\
    &\hspace{0.5cm}\prod_{i = 1}^N \Ex_{g \sim \sN(\bm 0, \|\bx\|^2)}h_{|\bm\alpha_i|}\left(g; -\frac{1 + \beta \|\bx\|^2}{\beta}\right)\exp\left(\frac{1}{2} \frac{\beta}{1 + \beta \|\bx\|^2} g^2 \right) \Bigg)^2.
    \label{eq:l-h-alpha-1}
\end{align}

We next focus on the innermost expectation. We may rewrite:
\begin{align}
    &\Ex_{g \sim \sN(\bm 0, \|\bx\|^2)}h_{|\bm\alpha_i|}\left(g; -\frac{1 + \beta \|\bx\|^2}{\beta}\right)\exp\left(\frac{1}{2} \frac{\beta}{1 + \beta \|\bx\|^2} g^2 \right) \nonumber \\
    &\hspace{0.5cm}= \frac{1}{\sqrt{2\pi \|\bx\|^2}}\int_{-\infty}^\infty h_{|\bm\alpha_i|}\left(g; -\frac{1 + \beta \|\bx\|^2}{\beta}\right)\exp\left(-\frac{1}{2}\left(\frac{1}{\|\bx\|^2} - \frac{\beta}{1 + \beta \|\bx\|^2}\right) g^2 \right)dg \nonumber \\
    &\hspace{0.5cm}= \frac{(1 + \beta\|\bx\|^2)^{1/2}}{\sqrt{2\pi \|\bx\|^2(1 + \beta\|\bx\|^2)}}\int_{-\infty}^\infty h_{|\bm\alpha_i|}\left(g; -\frac{1 + \beta \|\bx\|^2}{\beta}\right)\exp\left(-\frac{1}{2\|\bx\|^2(1 + \beta\|\bx\|^2)} g^2 \right)dg \nonumber \\
    &\hspace{0.5cm}= (1 + \beta\|\bx\|^2)^{1/2} \Ex_{g \sim \sN(0, \|\bx\|^2(1 + \beta\|\bx\|^2))}h_{|\bm\alpha_i|}\left(g; -\frac{1 + \beta \|\bx\|^2}{\beta}\right).
\end{align}
By Proposition~\ref{prop:hermite-mismatched}, this quantity will be zero unless $|\bm\alpha_i|$ is even, and thus $\la L_{n, \gamma, \beta, \sX}, \what{H}_{\bm\alpha}\ra^2$ will be zero unless $|\bm\alpha_i|$ is even for all $i$.
In this case, by Proposition~\ref{prop:hermite-mismatched},
\begin{equation}
    \Ex_{g \sim \sN(\bm 0, \|\bx\|^2)}h_{|\bm\alpha_i|}\left(g; -\frac{1 + \beta \|\bx\|^2}{\beta}\right)\exp\left(\frac{1}{2} \frac{\beta}{1 + \beta \|\bx\|^2} g^2 \right)
    = (|\bm\alpha_i| - 1)!!\frac{(1 + \beta\|\bx\|^2)^{|\bm\alpha_i| + 1/2}}{\beta^{|\bm\alpha_i| / 2}}.
\end{equation}
Substituting into~\eqref{eq:l-h-alpha-1}, we find many cancellations after which we are left with
\begin{equation}
    \la L_{n, \gamma, \beta, \sX}, \what{H}_{\bm\alpha} \ra^2 = \frac{\prod_{i = 1}^N (|\bm\alpha_i| - 1)!!^2}{\prod_{i = 1}^N \bm\alpha_i!} \beta^{|\bm\alpha|} \left(\Ex_{\bx \sim \sX_n} \bx^{\sum_{i = 1}^N \bm \alpha_i}\right)^2,
\end{equation}
the final result.
\end{proof}

\subsection{Norm of the Low-Degree Likelihood Ratio}

\begin{proof}[Proof of Lemma~\ref{lem:proj-norm}]
Recall that
\begin{equation}
    \|L^{\leq D}_{n, \gamma, \beta, \sX}\|_{L^2(\QQ_{n})}^2 = \sum_{\substack{\bm\alpha \in (\NN^n)^N \\ |\bm\alpha| \leq D}} \la L_{n, \gamma, \beta, \sX}, \what{H}_{\bm\alpha}\ra^2.
\end{equation}
We substitute in the result of Lemma~\ref{lem:proj-component-norms}, which, after introducing independent replicas $\bx^1, \bx^2 \sim \sX_n$, may be rewritten as
\begin{align*}
  \|L_{n, \gamma, \beta, \sX}^{\leq D}\|_{L^2(\QQ_{n})}^2
  &= \Ex_{\bx^1, \bx^2 \sim \sX_n}\sum_{\substack{\bm\alpha_i \in \NN^n, i \in [N] \\ |\bm\alpha_i| \text{even} \\ \sum_{i = 1}^N|\bm\alpha_i| \leq D}} \prod_{i = 1}^N \frac{(|\bm\alpha_i| - 1)!!^2}{\bm\alpha_i!} \beta^{|\bm\alpha_i|}(\bx^1)^{\bm\alpha_i}(\bx^2)^{\bm\alpha_i} \\
  &= \Ex_{\bx^1, \bx^2 \sim \sX_n}\sum_{d = 0}^D \beta^{d}\sum_{\substack{d_1, \dots, d_N \text{ even} \\ \sum d_i = d}}\left(\prod_{i = 1}^N \frac{(d_i - 1)!!^2}{d_i!}\right) \sum_{\substack{\bm\alpha_i \in \NN^n, i \in [N] \\ |\bm\alpha_i| = d_i}} \prod_{i = 1}^N \binom{d_i}{\bm\alpha_i}\prod_{j = 1}^n(x^1_jx^2_j)^{\bm\alpha_i(j)}.
\end{align*}
By the multinomial theorem,
\[ \la \bx^1, \bx^2 \ra^{d_i} = \sum_{\substack{\bm\alpha \in \NN^n \\ |\bm\alpha| = d_i}} \binom{d_i}{\bm\alpha}\prod_{j = 1}^n(x^1_jx^2_j)^{\bm\alpha(j)}, \]
and therefore
\begin{align*}
  \la \bx^1, \bx^2 \ra^{\sum_{i = 1}^N d_i}
  &= \prod_{i = 1}^N\la \bx^1, \bx^2 \ra^{d_i} \\
  &= \prod_{i = 1}^N \sum_{\substack{\bm\alpha \in \NN^n \\ |\bm\alpha| = d_i}} \binom{d_i}{\bm\alpha}\prod_{j = 1}^n(x^1_jx^2_j)^{\bm\alpha(j)} \\
  &= \sum_{\substack{\bm\alpha_i \in \NN^n, i \in [N] \\ |\bm\alpha_i| = d_i}} \prod_{i = 1}^N \binom{d_i}{\bm\alpha_i}\prod_{j = 1}^n(x^1_jx^2_j)^{\bm\alpha_i(j)}.
\end{align*}
In our case, this shows
\begin{align*}
  \|L_{n, \gamma, \beta, \sX}^{\leq D}\|_{L^2(\QQ_{n})}^2
  &= \Ex_{\bx^1, \bx^2 \sim \sX_n}\sum_{\substack{0 \leq d \leq D \\ d \text{ even}}} \beta^{d}\sum_{\substack{d_1, \dots, d_N \text{ even} \\ \sum d_i = d}}\left(\prod_{i = 1}^N \frac{(d_i - 1)!!^2}{d_i!}\right) \left(\prod_{i = 1}^N \la \bx^1, \bx^2 \ra^{d_i}\right) \\
  &= \Ex_{\bx^1, \bx^2 \sim \sX_n}\sum_{\substack{0 \leq d \leq D \\ d \text{ even}}} \beta^{d}\la \bx^1, \bx^2 \ra^d \sum_{\substack{d_1, \dots, d_N \text{ even} \\ \sum d_i = d}}\prod_{i = 1}^N \frac{d_i!}{d_i!!^2} \\
  &= \Ex_{\bx^1, \bx^2 \sim \sX_n}\sum_{\substack{0 \leq d \leq D \\ d \text{ even}}} 2^{-d}\beta^{d}\la \bx^1, \bx^2 \ra^d \sum_{\substack{d_1, \dots, d_N \text{ even} \\ \sum d_i = d}}\prod_{i = 1}^N \binom{d_i}{d_i / 2}
\end{align*}
where we have used the identities $n! = n!! \cdot (n - 1)!!$ and $(2n)!! = 2^n \cdot n!$.
We now pass to a notation making the restriction to even degrees clearer:
\[ \|L_{n, \gamma, \beta, \sX}^{\leq D}\|_{L^2(\QQ_{n})}^2 = \Ex_{\bx^1, \bx^2 \sim \sX_n}\sum_{0 \leq d \leq \lfloor D / 2 \rfloor } \left(\frac{\beta^2 \la \bx^1, \bx^2 \ra^2}{4}\right)^d \sum_{\substack{d_1, \dots, d_N \\ \sum d_i = d}}\prod_{i = 1}^N \binom{2d_i}{d_i}. \]
The remaining function may be understood in terms of the generating function of the central binomial coefficients: using Proposition~\ref{prop:gf-central-binom}, we have that for any $x \in (-\frac{1}{4}, \frac{1}{4})$,
\[ \varphi_N(x) \colonequals (1 - 4x)^{-N / 2} = \left(\sum_{d \geq 0} \binom{2d}{d}x^d\right)^N = \sum_{d \geq 0} x^d \sum_{\substack{d_1, \dots, d_N \\ \sum d_i = d}} \prod_{i = 1}^N \binom{2d_i}{d_i}. \]
Writing $\varphi_{N, k}(x)$ for the truncation of this Taylor series to degree $k$, we see that
\[ \|L_{n, \gamma, \beta, \sX}^{\leq D}\|_{L^2(\QQ_{n})}^2 = \Ex_{\bx^1, \bx^2 \sim \sX_n}\left[\varphi_{N, \lfloor D / 2 \rfloor}\left(\frac{\beta^2 \la \bx^1, \bx^2 \ra^2}{4}\right)\right], \]
the final result.
\end{proof}

\section{Proofs for Bounding the Low-Degree Likelihood Ratio}\label{app:low-deg}

\subsection{Local Chernoff Bound}
\label{app:local-chernoff}

\begin{proof}[Proof of Proposition~\ref{prop:local-chernoff}]
It is sufficient to show that $\iid(\pi/\sqrt{n})$ admits a local Chernoff bound. Since $\pi$ is subgaussian, $\pi^2 - \EE[\pi^2]$ is subexponential (see, e.g., \cite{rig-notes}), i.e., the moment-generating function $M(t) = \EE[\exp(t(\pi^2 - \EE[\pi^2]))]$ satisfies $M(t) \leq \exp(\frac{t^2}{2s^2})$ for all $|t| \leq s$ for a suitable choice of a constant $s > 0$.
In particular, $\EE[\exp(t \pi^2)] < \infty$ for all $|t| \le s$ for this choice of $s > 0$.

Let $\Pi = \pi \pi'$, the product of two independent copies of $\pi$. Let $\sigma^2$ be the variance proxy of $\pi$ (see Definition~\ref{def:subg}). The moment-generating function of $\Pi$ is
$$M(t) = \EE[\exp(t \Pi)] = \EE_\pi \EE_{\pi'}[\exp(t \pi \pi')] \le \EE_\pi\left[\exp\left(\sigma^2 t^2 \pi^2/2\right)\right] < \infty$$
provided $\frac{1}{2}\sigma^2 t^2 < s$, i.e., $|t| < \sqrt{2s/\sigma^2}$. Thus $M(t)$ exists in an open interval containing $t=0$, which implies $M^{\prime}(0) = \EE[\Pi] = 0$ and $M^{\prime\prime}(0) = \EE[\Pi^2] = 1$ (this is the defining property of the moment-generating function: its derivatives at $t = 0$ are the moments of $\Pi$).

Let $\eta > 0$ and $f(t) = \exp\left(\frac{t^2}{2(1-\eta)}\right)$. Since $M(0) = 1, M^{\prime}(0) = 0, M^{\prime\prime}(0) = 1$ and $f(0) = 1, f^{\prime}(0) = 0, f^{\prime\prime}(0) = \frac{1}{1-\eta} > 1$, there exists $\delta > 0$ such that for all $t \in [-\delta,\delta]$, $M(t)$ exists and $M(t) \le f(t)$.

We now apply the standard Chernoff bound argument to $\la \bx^1,\bx^2 \ra = \frac{1}{n} \sum_{i=1}^n \Pi_i$, where $\Pi_1,\ldots,\Pi_n$ are i.i.d.\ copies of $\Pi$. For any $\lambda > 0$,
\begin{align*}
\Pr\left\{\la \bx^1,\bx^2 \ra \ge t\right\} &= \Pr\left\{\exp(\lambda \la \bx^1,\bx^2 \ra) \ge \exp(\lambda t)\right\}\\
&\le \exp(-\lambda t) \EE[\exp(\lambda \la \bx^1,\bx^2 \ra)] \tag{by Markov's inequality}\\
&= \exp(-\lambda t) \EE[\exp(\lambda n^{-1} \sum_{i=1}^n \Pi_i)]\\
&= \exp(-\lambda t) [M(\lambda/n)]^n\\
&\le \exp(-\lambda t) [f(\lambda/n)]^n \tag{provided $\lambda/n \le \delta$} \\
&\le \exp(-\lambda t) \exp\left(\frac{\lambda^2}{2(1-\eta)n}\right).
\end{align*}
Taking $\lambda = (1-\eta)nt$,
$$\Pr\left\{\la \bx^1,\bx^2 \ra \ge t\right\} \le \exp\left(-(1-\eta)nt^2 + \frac{1}{2}(1-\eta)nt^2\right) = \exp\left(-\frac{1}{2}(1-\eta)nt^2\right)$$
as desired. This holds provided $\lambda/n \le \delta$, i.e., $t \le \delta/(1-\eta)$. The same argument (with $-\Pi$ instead of $\Pi$) holds for the other tail bound $\Pr\left\{\la \bx^1,\bx^2 \ra \le -t\right\}$.
\end{proof}

\subsection{Bounding the Large Deviations}
\label{app:large-dev}

\begin{proof}[Proof of Lemma~\ref{lem:large-dev}]
Recall that
\[ \varphi_{N, \lfloor D/2 \rfloor}(x) = \sum_{d=0}^{\lfloor D/2 \rfloor} x^d \sum_{\substack{d_1, \dots, d_N \\ \sum d_i = d}} \prod_{i = 1}^N \binom{2d_i}{d_i}. \]
Note that the first sum above has $\lfloor D/2 \rfloor+1$ terms and the second sum has at most $N^d \le N^{\lfloor D/2 \rfloor} \leq N^{D / 2}$ terms. It is combinatorially clear that $\binom{2d_i}{d_i}\binom{2d_j}{d_j} \le \binom{2(d_i+d_j)}{d_i+d_j}$, and therefore
\[ \prod_{i=1}^N \binom{2d_i}{d_i} \le \binom{2 \sum_{i = 1}^N d_i}{\sum_{i = 1}^N d_i} = \binom{2d}{d} \le (2d)^d \le D^{D/2}. \]
Since $\|\bx^1\|^2 \le 2$ and $\|\bx^2\|^2 \le 2$ we have $\frac{1}{4}\beta^2 \la \bx^1, \bx^2 \ra^2 \le \beta^2$. Since $d \le D/2$ we have $(\beta^2)^d \le (1+\beta^2)^{D/2}$, and therefore
\[ \varphi_{N, \lfloor D/2 \rfloor}\left(\frac{\beta^2}{4} \la \bx^1, \bx^2 \ra^2\right) \le (D/2+1) (1+\beta^2)^{D/2} N^{D/2} D^{D/2} \le (1+\beta^2)^{D/2} D N^{D/2} D^{D/2}. \]

Combining these bounds,
\begin{align*}
R_2 &= \Ex_{\bx^1, \bx^2 \sim \sX_n}\left[\One_{|\langle \bx^1,\bx^2 \rangle| > \varepsilon}\, \varphi_{N, \lfloor D/2 \rfloor}\left(\frac{\beta^2}{4} \la \bx^1, \bx^2 \ra^2\right)\right]\\
&\le \Pr\left\{|\langle \bx^1,\bx^2 \rangle| > \varepsilon\right\} (1+\beta^2)^{D/2} D N^{D/2} D^{D/2}.
\intertext{Since $R_2$ increases as $\varepsilon$ decreases, we can assume without loss of generality that $\varepsilon$ is small enough that we may apply the local Chernoff bound \eqref{eq:local-chernoff}:}
&\le \exp\left(-\frac{1}{3}n\varepsilon^2\right) (1+\beta^2)^{D/2} D N^{D/2} D^{D/2}\\
&= \exp\left(-\frac{1}{3}n\varepsilon^2 + \frac{D}{2} \log(1+\beta^2) + \log D + \frac{D}{2} \log N + \frac{D}{2} \log D\right)\\
&= o(1)
\end{align*}
provided $D = o(n/\log n)$, completing the proof.
\end{proof}

\subsection{Bounding the Small Deviations}
\label{app:small-dev}

\begin{proof}[Proof of Lemma~\ref{lem:small-dev}]
We use the argument from Appendix~K of \cite{PWBM-contig}. Since the Taylor series for $\varphi_N(x)$ has nonnegative coefficients, we have $\varphi_{N, \lfloor D/2 \rfloor}(x) \le \varphi_N(x)$ for all $x \in [0,1/4)$. Taking $\varepsilon < 1/|\beta|$, we have
\begin{align*}
R_1 &= \Ex_{\bx^1, \bx^2 \sim \sX_n}\left[\One_{|\langle \bx^1,\bx^2 \rangle| \le \varepsilon}\, \varphi_{N, \lfloor D/2 \rfloor}\left(\frac{\beta^2}{4} \la \bx^1, \bx^2 \ra^2\right)\right]\\
&\le \Ex_{\bx^1, \bx^2 \sim \sX_n}\left[\One_{|\langle \bx^1,\bx^2 \rangle| \le \varepsilon}\, \varphi_N\left(\frac{\beta^2}{4} \la \bx^1, \bx^2 \ra^2\right)\right]\\
&= \Ex_{\bx^1, \bx^2 \sim \sX_n}\left[\One_{|\langle \bx^1,\bx^2 \rangle| \le \varepsilon}\, \left(1-\beta^2 \la \bx^1, \bx^2 \ra^2\right)^{-N/2}\right]\\
&= \Ex_{\bx^1, \bx^2 \sim \sX_n}\left[\One_{|\langle \bx^1,\bx^2 \rangle| \le \varepsilon}\, \exp\left(-\frac{N}{2} \log\left(1-\beta^2 \la \bx^1, \bx^2 \ra^2\right)\right)\right].\\
\intertext{By the convexity of $t \mapsto -\log(1-\beta^2 t)$, we have $-\log(1-\beta^2 t) \le -\frac{t}{\varepsilon^2}(1-\beta^2 \varepsilon^2)$ for all $t \in [0,\varepsilon^2]$. Letting $c \colonequals -\frac{N}{2\varepsilon^2} \log(1-\beta^2\varepsilon^2) > 0$, we proceed bounding}
&\le \Ex_{\bx^1, \bx^2 \sim \sX_n}\left[\One_{|\la \bx^1,\bx^2 \ra| \le \varepsilon}\, \exp\left(c\la \bx^1, \bx^2 \ra^2\right)\right]\\
&= \int_0^\infty \Pr\left\{\One_{|\la \bx^1,\bx^2 \ra| \le \varepsilon}\, \exp\left(c\la \bx^1, \bx^2 \ra^2\right) > u\right\} du\\
&= \int_0^\infty \Pr\left\{|\la \bx^1,\bx^2 \ra| \le \varepsilon \quad\text{and}\quad \exp\left(c\la \bx^1, \bx^2 \ra^2\right) > u\right\} du\\
&\le 1 + \int_1^\infty \Pr\left\{|\la \bx^1,\bx^2 \ra| \le \varepsilon \quad\text{and}\quad \exp\left(c\la \bx^1, \bx^2 \ra^2\right) > u\right\} du.
\intertext{Applying the change of variables $u = \exp\left(ct\right)$,}
&= 1 + \int_0^\infty \Pr\left\{|\la \bx^1,\bx^2 \ra| \le \varepsilon \quad\text{and}\quad \la \bx^1, \bx^2 \ra^2 > t\right\}c\exp\left(ct\right) dt\\
&\le 1 + \int_0^{\varepsilon^2} \Pr\left\{\la \bx^1, \bx^2 \ra^2 > t\right\}c\exp\left(ct\right) dt.\\
\intertext{Provided $\varepsilon$ is sufficiently small, we can apply the local Chernoff bound (\ref{eq:local-chernoff}):}
&\le 1 + C c \int_0^{\varepsilon^2} \exp\left(-\frac{1}{2}(1-\eta)nt + ct\right) dt.\\
\intertext{Let $\what{\gamma} \colonequals n/N$, so that $\what\gamma \to \gamma$ as $n \to \infty$. Letting $c \equalscolon \what{c} n$ where $\what{c} = -\log(1-\beta^2\varepsilon^2) / (2\varepsilon^2 \what\gamma)$,}
&\le 1 +  C \cdot \what c n \int_0^{\varepsilon^2} \exp\left[\left(-\frac{1}{2}(1-\eta) + \what c\right)nt\right] dt.\\
\intertext{We have $\lim_{\varepsilon \to 0^+} \what c = \frac{\beta^2}{2\what\gamma}$. Since $\beta^2 < \gamma$, we have that for sufficiently large $n$, $\frac{\beta^2}{2\what\gamma} < \frac{1}{2}$. Thus we can choose $\varepsilon$ and $\eta$ small enough so that for sufficiently large $n$, $-\frac{1}{2}(1-\eta)+\what c \le -\alpha$ for some $\alpha > 0$. Now}
&\le 1 + C \cdot \what c n \int_0^\infty \exp\left(-\alpha nt\right) dt\\
&= 1 + \frac{C \cdot \what c}{\alpha}\\
&= O(1),
\end{align*}
completing the proof.
\end{proof}

\subsection{Above the BBP Threshold}
\label{app:above-thresh}

\begin{proof}[Proof of Theorem~\ref{thm:wish-lowdeg} (Part 2)]
Let $\beta^2 > \gamma$. Recall

\begin{align}
\|L_{n, \gamma, \beta, \sX}^{\leq D}\|_{L^2(\QQ_{n})}^2
&= \Ex_{\bx^1, \bx^2 \sim \sX_n}\left[\varphi_{N, \lfloor D/2 \rfloor}\left(\frac{\beta^2}{4} \la \bx^1, \bx^2 \ra^2\right)\right] \nonumber \\
&= \sum_{d=0}^{\lfloor D/2 \rfloor} \Ex_{\bx^1, \bx^2 \sim \sX_n}\left(\frac{\beta^2}{4} \la \bx^1, \bx^2 \ra^2\right)^d \sum_{\substack{d_1, \dots, d_N \\ \sum d_i = d}} \prod_{i = 1}^N \binom{2d_i}{d_i}. \label{eq:above-thresh-expansion}
\end{align}

Since each term in the outer summation of \eqref{eq:above-thresh-expansion} is nonnegative, it sufficies to fix a single $d \le D/2$ and show that the corresponding term is $\omega(1)$. We can write $\la \bx^1,\bx^2 \ra = \frac{1}{n} \sum_{i=1}^n \Pi_i$ where $\Pi_1,\ldots,\Pi_n$ are i.i.d.\ with distribution of the product $\Pi = \pi \pi'$ of two independent copies of $\pi$. This means
$$\Ex_{\bx^1, \bx^2 \sim \sX_n} \la \bx^1, \bx^2 \ra^{2d} = \EE \left(\frac{1}{n} \sum_{i=1}^n \Pi_i\right)^{2d} = n^{-2d} \sum_{i_1,\ldots,i_{2d} \in [n]} \EE[\Pi_{i_1} \Pi_{i_2} \cdots \Pi_{i_{d2}}].$$
Since $\pi$ is symmetric about zero, $\Pi$ is also symmetric about zero, so all moments $\EE[\Pi^k]$ are nonnegative. This means each term in the remaining sum is nonnegative, so we can obtain a lower bound by only considering terms where each index occurring among the $i_1, \dots, i_{2d}$ occurs exactly twice:
$$\Ex_{\bx^1, \bx^2 \sim \sX_n} \la \bx^1, \bx^2 \ra^{2d} \ge n^{-2d} \binom{n}{d} \frac{(2d)!}{2^d} \left(\EE[\Pi^2]\right)^d = n^{-2d} \binom{n}{d} \frac{(2d)!}{2^d}.$$
Next, we bound the inner summation of \eqref{eq:above-thresh-expansion} below by taking only the terms with $d_i \in \{0,1\}$ for all $i \in [N]$:
$$\sum_{\substack{d_1, \dots, d_N \\ \sum d_i = d}} \prod_{i = 1}^N \binom{2d_i}{d_i} \ge \binom{N}{d} 2^d.$$
Combining these bounds, we find that for any fixed $0 \leq d \leq \lfloor D / 2 \rfloor$,
\begin{align*}
\|L_{n, \gamma, \beta, \sX}^{\leq D}\|_{L^2(\QQ_{n})}^2
&\ge \frac{\beta^{2d}}{4^d} n^{-2d} \binom{n}{d} \frac{(2d)!}{2^d} \binom{N}{d} 2^d \\
&= \left(\frac{\beta^2}{4n^2}\right)^d \frac{(2d)!\, n!\, N!}{(d!)^2 \,(n-d)!\, (N-d)!} \\
&\ge \left(\frac{\beta^2 (n-d)(N-d)}{4n^2}\right)^d \binom{2d}{d}.
\end{align*}
Using the standard bound $\binom{2d}{d} \ge 4^d / (2\sqrt{d})$,


$$\|L_{n, \gamma, \beta, \sX}^{\leq D}\|_{L^2(\QQ_{n})}^2 \ge \frac{1}{2\sqrt{d}} \left(\frac{\beta^2 (n-d)(N-d)}{n^2}\right)^d.$$

\noindent This final expression will be $\omega(1)$ provided that $1 \ll d \ll n$, since $n/N \to \gamma$ and $\beta^2 > \gamma$.
\end{proof}

\end{document}